\documentclass{article}\usepackage[]{graphicx}\usepackage[]{color}
\makeatletter
\def\maxwidth{ %
  \ifdim\Gin@nat@width>\linewidth
    \linewidth
  \else
    \Gin@nat@width
  \fi
}
\makeatother

\definecolor{fgcolor}{rgb}{0.345, 0.345, 0.345}

\usepackage{framed}
\makeatletter
 {\par\unskip\endMakeFramed%
 \at@end@of@kframe}
\makeatother

\definecolor{shadecolor}{rgb}{.97, .97, .97}
\definecolor{messagecolor}{rgb}{0, 0, 0}
\definecolor{warningcolor}{rgb}{1, 0, 1}
\definecolor{errorcolor}{rgb}{1, 0, 0}
\newenvironment{knitrout}{}{} 

\usepackage{alltt}
\usepackage{url}
\usepackage{bbm}
\usepackage{natbib}
\usepackage{amsmath,amsthm,amssymb,enumerate,enumitem}
\usepackage{thmtools} 
\usepackage{multirow,array}
\usepackage{authblk}
\usepackage[a4paper, total={6in, 9in}]{geometry}

\newcommand{\E}{\mathbb{E}}

\newcommand{\R}{\mathbb{R}}
\renewcommand{\P}{\mathbb{P}}

\newcommand\independent{\protect\mathpalette{\protect\independenT}{\perp}}
\def\independenT#1#2{\mathrel{\rlap{$#1#2$}\mkern2mu{#1#2}}}
\setlist[1]{itemsep=-7pt}

\declaretheorem[style=theorem,numberwithin=section]{theorem}
\declaretheorem[style=definition,numberwithin=section]{definition}
\declaretheorem[style=definition,numberwithin=section]{procedure}

\title{Faster Family-wise Error Control for Neuroimaging with a Parametric Bootstrap}

\author[1]{Simon N. Vandekar}
\author[2]{Theodore D. Satterthwaite}
\author[2]{Adon Rosen}
\author[2]{Rastko Ciric}
\author[2]{David R. Roalf}
\author[2]{Kosha Ruparel}
\author[2,3]{Ruben C. Gur}
\author[2,3]{Raquel E. Gur}
\author[1]{Russell T. Shinohara}
\affil[1]{Department of Biostatistics, Epidemiology, and Informatics, University of Pennsylvania, Philadelphia PA}
\affil[2]{Department of Psychiatry, University of Pennsylvania, Philadelphia PA}
\affil[3]{Department of Radiology, University of Pennsylvania, Philadelphia PA}
\affil[ ]{{\tt  simonv@mail.med.upenn.edu; rshi@mail.med.upenn.edu} }
\date{}                     
\IfFileExists{upquote.sty}{\usepackage{upquote}}{}
\begin{document}



\maketitle



\begin{abstract}
{In neuroimaging, hundreds to hundreds of thousands of tests are performed across a set of brain regions or all locations in an image.
Recent studies have shown that the most common family-wise error (FWE) controlling procedures in imaging, which rely on classical mathematical inequalities or Gaussian random field theory, yield FWE rates that are far from the nominal level.
Depending on the approach used, the FWER can be exceedingly small or grossly inflated.
Given the widespread use of neuroimaging as a tool for understanding neurological and psychiatric disorders, it is imperative that reliable multiple testing procedures are available.
To our knowledge, only permutation joint testing procedures have been shown to reliably control the FWER at the nominal level.
However, these procedures are computationally intensive due to the increasingly available large sample sizes and dimensionality of the images, and analyses can take days to complete.
Here, we develop a parametric bootstrap joint testing procedure. The parametric bootstrap procedure works directly with the test statistics, which leads to much faster estimation of adjusted \emph{p}-values than resampling-based procedures while reliably controlling the FWER in sample sizes available in many neuroimaging studies.
We demonstrate that the procedure controls the FWER in finite samples using simulations, and present region- and voxel-wise analyses to test for sex differences in developmental trajectories of cerebral blood flow.}
{Hypothesis testing; FWE control; Neuroimaging}
\end{abstract}

\section{Introduction}

 Magnetic resonance imaging (MRI) is a widely used tool for studying the neurological correlates of human cognition, psychiatric disorders, and neurological diseases.
 This is due to the flexibility of MRI to noninvasively study various functional and physiological properties of the human brain.
 Often, many hypothesis tests are performed at every voxel or at anatomically defined brain regions in order to identify locations that are associated with a cognitive or diagnostic variable.
 Multiple testing procedures (MTPs) are crucial for controlling the number of false positive findings within a statistical parametric map or across a set of brain regions being investigated.
Typically the family-wise error rate (FWER) is controlled at a level $0<\alpha<1$, meaning that the probability one or more null hypotheses is falsely rejected is less than or equal to $\alpha$.

Recently, several studies have demonstrated that commonly used FWER controlling procedures yield incorrect false positive rates \citep{eklund_cluster_2016,eklund_does_2012,silver_false_2011}.
Cluster-based spatial inference procedures \citep{friston_statistical_1994} that rely on Gaussian random field (GRF) theory can have hugely inflated false positive rates, while voxel-wise GRF MTPs  \citep{friston_assessing_1994} tend to have exceedingly small FWERs that are far below the nominal level.
The failure of GRF procedures is due to the fact that the spatial assumptions of Gaussian random field approaches are often violated in neuroimaging data sets \citep{eklund_cluster_2016}.
The small type 1 error rate of voxel-wise procedures is due to the reliance on classical FWER procedures (such as the Bonferroni procedure) that do not account for the strong dependence between hypothesis tests in voxel-wise and region-wise analyses.
This small type 1 error rate leads to an inflated type 2 error rate and loss of power.
These recent studies demonstrate a dire need for robust and powerful inference procedures.

To our knowledge, the only methods used in neuroimaging that reliably control the FWER are permutation-based joint testing procedures \citep{winkler_permutation_2014,eklund_cluster_2016,dudoit_multiple_2008}.
These methods maintain the nominal FWER because they appropriately reproduce the joint distribution of the imaging data, thereby overcoming the limitations of methods typically used in imaging.
Unfortunately, permutation testing is computationally intensive, especially in modern imaging data sets that have large sample sizes and hundreds of thousands of voxels or hundreds of brain regions.
Moreover, currently available neuroimaging software only performs single-step testing procedures, although step-down procedures are uniformly more powerful.
The extensive computation time means it can take days to perform statistical tests, or can even lead investigators to reduce the number of permutations to an extent that adjusted \emph{p}-values have large error.

As a solution, we design a parametric-bootstrap joint (PBJ) testing procedure for hypothesis testing in region- and voxel-wise analyses.
Region-wise analyses are performed by averaging all voxel values within anatomically defined regions and fitting a model at each region.
Voxel-wise analyses fit a model at each of hundreds of thousands of voxels in a brain image.
As the parametric bootstrap does not require resampling and refitting the model for every iteration, it is faster than permutation testing procedures.
In addition, our procedure allows the generated null distribution to be applied to multiple tests of statistical parameters.
This drastically reduces computing time as the null distribution can be estimated from one bootstrap procedure and applied for many tests.
We demonstrate the efficacy of our procedure by investigating sex differences in development-related changes of cerebral blood flow (CBF) measured using arterial spin labeled MRI \citep{satterthwaite_impact_2014}.

All joint testing procedures rely on an estimate of the joint null distribution of the test statistics.
Therefor, all joint testing procedures produce approximate \emph{p}-values that only guarantee asymptotic FWE control as the sample size goes to infinity.
However, permutation procedures are widely used in genetics and neuroimaging when the number of tests grossly exceeds the sample size \citep{westfall_multiple_2008,winkler_permutation_2014}.
A recent neuroimaging simulation study demonstrated that permutation tests control the FWER at the nominal level even when the number of tests exceeds the sample size \citep{eklund_cluster_2016}.
To investigate this feature further, we perform a simulation study to investigate when joint MTPs control the FWER.

In Section \ref{sec:mtps} we discuss several FWER controlling procedures used in neuroimaging and classify them with regard to single-step/step-down and marginal/ joint procedures.
In Section \ref{sec:jointdist} we present the new PBJ procedure.
We summarize the data set and simulation methods in Section \ref{sec:methods}.
In Section \ref{sec:Simulations} we use simulations to investigate when joint MTPs maintain the nominal type 1 error rate, and we compare the power and FWER of the PBJ to commonly used MTPs using simulations of region- and voxel-wise data analyses.
Finally, in Section \ref{sec:results} we perform region- and voxel-wise analyses of the CBF data.
 

\section{Overview of Multiple Testing Procedures}
\label{sec:mtps}

Throughout, we will assume the image intensity for $n$ subjects, $Y_v \in \R^n$, for voxels or regions, $v =1 \ldots, V$, can be expressed as the linear model
\begin{equation}\label{eq:model}
Y_v = X_0 \alpha_v + X_1 \beta_v + \epsilon_v = X \zeta_v + \epsilon_v,
\end{equation}
where $X_0$ is an $n \times m_0$ matrix of nuisance covariates, $X_1$ is an $n \times m_1$ matrix of variables to be tested, $m=m_0 + m_1$, $X = [ X_0, X_1]$, parameters $\alpha_v \in \R^{m_0}$, $\beta_v \in \R^{m_1}$, and $\zeta_v = [\alpha_v^T, \beta_v^T ]^T$.
Let $Y = [Y_1, \ldots, Y_V]$ and let $Y_i$ denote an arbitrary row vector of $Y$.
Assume that the $V\times V$ covariance matrix is the same for each subject, $\text{cov}(Y_i) = \Psi$, and define the correlation matrix 
\begin{equation}
\label{eq:sigma}
\Sigma_{j,k} = \Psi_{j,k}/\sqrt{\Psi_{j,j} \Psi_{k,k}}.
\end{equation}
We denote the observed test statistics by $Z_{v0}$ for $v=1,\ldots, V$, where we reject the null $H_0:\beta_v = 0$ for large values of $Z_{v0}$.
The notation $Z$ is used (as opposed to $F$) as we consider transformed F-statistics in Section \ref{sec:jointdist}.

At each location we are interested in performing the test of the null hypothesis 
\[
H_{0v}: \beta_v = 0
\]
using an F-statistic.
The form of model \eqref{eq:model} covers a wide-range of possible tests including tests of group differences, continuous covariates, analysis of variance, and interactions.
$V$ is typically in the hundreds for region-wise analyses or the hundreds of thousands for voxel-wise analyses.

The goal of all multiple testing procedures is to control some measure of the number of false positive findings in a family of hypothesis tests.
We will assume the approach of controlling the FWER at some level $0<\alpha <1$.
In most fields we would like to maintain control of the FWER even in the case that there are false null hypotheses. 
This is referred to as strong control of the FWER \citep{hochberg_sharper_1988}.
\begin{definition}
Let $\{ H_1, \ldots, H_{V} \} = H$ denote a set of hypotheses.
A correction procedure has $\alpha$ level {\it strong control of the FWER} if for all $H' \subset H$ of true null hypotheses
\begin{equation}\label{eq:type1error}
\P(\text{retain $H_v$ for all $H_v\in H'$}) \ge 1 - \alpha.
\end{equation}
\end{definition}
In neuroimaging, strong control of the FWER corresponds to maintaining the correct FWER control even if there is a set of regions or voxels where there is a true effect.
The strong FWER controlling procedures discussed in this manuscript are given in Table \ref{tab:methods}.

\begin{table}
\centering
\begin{tabular}{r|llll}
Procedure & Analysis & Marg/Joint & Null estimation & Step proc \\
\hline
Bonferroni & Region-wise & Marg & Theor & Single-step  \\
Holm  & Region-wise& Marg & Theor & Step-down   \\
PBJ  & Region-wise& Joint  & Theor; Boot & Single-step/Step-down  \\
Permutation & Region-wise & Joint & Boot & Single-step/Step-down  \\
\hline
Bonferroni & Voxel-wise & Marg & Theor & Single-step \\
Holm & Voxel-wise & Marg & Theor & Step-down  \\
PBJ & Voxel-wise & Joint & Theor; Boot & Single-step   \\
Permutation & Voxel-wise & Joint & Perm & Single-step   \\
\end{tabular}
\caption{A summary of hypothesis correction procedures for neuroimaging. Joint methods are more powerful than marginal methods. Step-down procedures are more powerful than single-step. Marg$=$ marginal; Step proc$= $ step procedure; Speed$=$ computing speed; Theor$=$ Assumes theoretical distribution; Perm$=$ distribution obtain using permutations.}
\label{tab:methods}
\end{table}

\subsection{Single-step and Step-down procedures}
Due to the complex dependence structure of the test statistics in neuroimaging data it is necessary to use testing procedures that are appropriate for any type of dependence amongst tests.
Single-step and step-down procedures are two classes of MTPs that have strong control of the FWER for any dependence structure \citep{dudoit_multiple_2003}.
These are in contrast to step-up procedures, which make explicit assumptions about the dependence structure of the test statistics \citep{sarkar_simes_1997}.
When testing $V$ hypotheses $H = \{H_1, \ldots, H_V\}$ in a family of tests, the single-step procedures use a more stringent common threshold $\alpha^* \le \alpha$ such that the inequality \eqref{eq:type1error} is guaranteed.
While this procedure is simple, in most cases it is uniformly more powerful to use a step-down procedure.

\begin{procedure}[Step-down procedure]
\label{proc:stepdown}
Let $p_{(1)}, \ldots, p_{(V)}$ be the increasingly ordered \emph{p}-values, for hypotheses $H_{(1)}, \ldots, H_{(V)}$.
The step-down procedure uses thresholds $\alpha^*_{(1)} \le \ldots \le \alpha^*_{(V)} \le \alpha$ to find the largest value of $K$ such that
\begin{align*}
p_{(k)} < \alpha^*_{(k)} & \text{ for all $k\le K$},
\end{align*}
 and rejects all hypotheses $H_{(1)}, \ldots H_{(K)}$.
\end{procedure}
 
 The single-step procedure can usually be improved by modifying it to be step-down procedure while still maintaining strong control of the FWER \citep{holm_simple_1979,marcus_closed_1976,hommel_stagewise_1988}.
 The canonical example of this modification involves the Bonferroni and Holm procedures \citep{dunn_multiple_1961,holm_simple_1979}.
 The Bonferroni procedure uses the common threshold $\alpha^* = \alpha/V$ for all hypotheses, and rejects all $H_{k}$ such that $p_k < \alpha/V$.
 The Holm procedure instead uses the thresholds $\alpha^*_{(k)} = \alpha/(V +1 - k)$ and rejects using Procedure \ref{proc:stepdown}.
 Holm's procedure is always more powerful than Bonferroni and still controls the FWER strongly \citep{holm_simple_1979}.


\subsection{Joint testing procedures}

Multiple testing procedures can further be classified into marginal and joint testing procedures.
The Bonferroni and Holm approaches are called marginal procedures because they do not make use of the dependence of the test statistics.
As they must be able to control any dependence structure they are more conservative than joint testing procedures that cater exactly to the distribution of the test statistics \citep{dudoit_multiple_2008}.
The benefit of accounting for the dependence of test statistics is critical in neuroimaging, where the test statistics are highly dependent due to spatial, anatomical, and functional dependence.
The joint MTPs differ in how the null distribution is estimated from the sample, but use the same procedure to compute single-step or step-down adjusted \emph{p}-values.
For this reason we will first discuss the estimation of the null distribution and then discuss how adjusted \emph{p}-values are computed from the estimate of the null.

\subsubsection{Estimating the null distribution with permutations}
The ``Randomise" procedure proposed by \citet{winkler_permutation_2014} is a single-step permutation joint (PJ) MTP widely used in neuroimaging.
The PJ MTP procedure is a modification of the Freedman-Lane procedure \citep{freedman_nonstochastic_1983} implemented by \citet{winkler_permutation_2014} that estimates the null distribution of the test statistics by permuting the residuals of the reduced model to obtain estimates of the parameters under the null hypothesis that there is no association between the variables in $X_1$ and the outcome.
Though only the single-step procedure has been proposed for use in neuroimaging, for completeness we include null estimation of the test statistics for the step-down procedure.
\begin{procedure}[Permutation Null Estimation]
Assuming the model \eqref{eq:model}:
\label{proc:permutation}
 \begin{enumerate}
 \item Regress $Y_v$ against the reduced model $Y_v = X_0 \alpha_v + \epsilon_v$ to obtain the residuals $\hat \epsilon_{v}$ and the test statistics $Z_{v0}$ for all regions or voxels $v = 1, \ldots, V$.
 \item Order the test statistics $Z_{(1)0} < Z_{(2)0} < \ldots Z_{(V)0}$ and let $\epsilon_{(v)}$ be the corresponding residuals.
 \item For $b=1, \ldots, B$, randomly generate a permutation matrix $P_b$, permute the residuals $\hat \epsilon_{(v)b} = P_b\hat\epsilon_{(v)}$, and define the permuted data at each voxel as $Y_{(v)b} = \hat \epsilon_{(v)b}$.
 \item For $v=1, \ldots, V$ and $b=1, \ldots, B$ regress $Y_{(v)b}$ onto the full model \eqref{eq:model} to obtain the test statistic $Z_{(v)b}$ to be used as an estimate of the null distribution.
 \end{enumerate}
\end{procedure}
Ordering the test statistics is not necessary for the single-step procedure, but is required for computing step-down adjusted \emph{p}-values.
Note that for any given $b$ the generated test statistics $Z_{(v)b}$ may not be increasing in $v$.
Strong control of the FWER for this permutation procedure relies on the assumption of subset pivotality (see Supplement) \citep{westfall_resampling-based_1993}.

This null distribution can be used to compute rejection regions or adjusted \emph{p}-values.
Because all joint testing procedures rely on estimates of the null distribution they are approximate in finite samples.
These procedures only guarantee asymptotic control of the FWER as $n\to\infty$.
The permutation methodology and our proposed PBJ procedure (Section \ref{sec:jointdist}) only differ in how the null distribution is estimated.

\subsubsection{Computing adjusted p-values}

The following procedures describe how to obtain single-step and step-down adjusted \emph{p}-values using any estimate of the null distribution generated by permutation or bootstrapping.

\begin{procedure}[Single-Step Joint Adjusted \emph{p}-values]
\label{proc:bootSS}
Assuming the model \eqref{eq:model} and an empirical distribution of null statistics $Z_{(v)b}$ for $v=1, \ldots, V$ and $b=1, \ldots, B$:
 \begin{enumerate}
 \item Compute $Z_{\text{max},b} = \max_{v\le V} Z_{(v)b}$.
 \item Compute the voxel-wise corrected \emph{p}-value as $\tilde p_{v} = 1/B \sum_{b=1}^b I(Z_{\text{max},b}\ge Z_{v0})$, where $I(\cdot)$ is the indicator function.
 \end{enumerate}
\end{procedure}

\begin{procedure}[Step-down Joint Adjusted \emph{p}-values]
\label{proc:bootSD}
Assuming the model \eqref{eq:model} and an empirical distribution of null statistics $Z_{(v)b}$ for $v=1, \ldots, V$ and $b=1, \ldots, B$:
 \begin{enumerate}
 \item Compute the statistics $Z_{\max v,b} = \max_{k\le v} Z_{(k)b}$.
 \item Compute the the intermediate value $ p^*_{(v)} = 1/B \sum_{b=1}^B I(Z_{\max v,b}\ge Z_{(v)0})$.
 \item The adjusted \emph{p}-values are $\tilde p_{(v)} = \max_{k\le v} p^*_{(k)}$.
 \end{enumerate}
\end{procedure}

The single-step procedure is less powerful than the step-down counterpart \citep{dudoit_multiple_2008}.
This is evident in comparing the procedures, as the adjusted \emph{p}-values for the step-down procedure are at least as small as the adjusted \emph{p}-values from the single-step.
The adjusted \emph{p}-values obtained using the step-down approach correspond to using Procedure \ref{proc:stepdown}.
They key feature of Procedure \ref{proc:bootSS} is that the estimate of the joint distribution is used to compute adjusted \emph{p}-values, where as Holm's procedure is a version of Procedure \ref{proc:stepdown} that only uses the marginal distribution of the test statistics. Up to this point we have described existing MTPs used in neuroimaging.


\section{Parametric-Bootstrap}
\label{sec:jointdist}

In this section we propose single-step and step-down PBJ approaches that are conceptually identical to the PJ procedure, but differ in how the null distribution of the statistics is generated.
The PBJ is based on the theory developed by \citet{dudoit_multiple_2008} and therefore does not rely on the assumption of subset pivotality.
We will allow the additional assumption that under the null the test statistics are approximately chi-squared.
The chi-squared approximation can rely on asymptotic results, or as we show in Section \ref{sec:CBFmethods}, a transformation can be used so that the test statistics are approximately chi-squared.
As with the PJ procedure discussed above this implies that the \emph{p}-values are approximations that become more accurate as $n \to \infty$.
In Section \ref{sec:Simulations} we use simulations to show that the procedures control the FWER in sample sizes available in many neuroimaging studies.

\subsection{Asymptotic control of the FWER}

Here, we give a brief overview of the underlying assumptions sufficient to prove that the adjusted \emph{p}-values from the PBJ control the FWER asymptotically.
Details are given in the Supplement.
We require that the test statistics' null distribution satisfies the null domination condition \citep[p.~203]{dudoit_multiple_2008} and need a consistent estimate of the null distribution.

\begin{definition}[Asymptotic null domination]
\label{def:and}
Let $H_0$ denote the indices of $M$ true null hypotheses in the set of $V$ hypotheses $H = \{ H_1, \ldots, H_V\}$, with corresponding test statistics $Z_{1n}, \ldots, Z_{Vn}$.
The $V$-dimensional null distribution $Q_0$ satisfies the {\it asymptotic null domination} condition if for all $x \in \R$
\[
\limsup_{n\to \infty} \P\left( \max_{m \in H_0} Z_{mn} > x  \right)
\le \P\left(\max_{m\in H_0} Z_{m} > x \right),
\]
where $Z_n \sim Q_n$ is distributed according to a finite sample null joint distribution $Q_n$ and $Z \sim Q_0$.
\end{definition}

The joint null distribution $Q_0$ for the test statistics can be used to compute asymptotically accurate adjusted \emph{p}-values if the null domination condition holds.
We use a diagonal singular Wishart distribution as the null because it is proportional to the asymptotic distribution a vector of F-statistics. We also transform the F-statistics marginally to chi-squared statistics if
\begin{equation}
\label{eq:normality}
\epsilon_v \sim \mathcal{N}(0, \sigma_v I),
\end{equation}
for the error term in model \eqref{eq:model}.
The assumption of asymptotic null domination of Definition \ref{def:and} is satisfied when using our transformated F-statistics even if the error distribution is not normal (See Theorem \ref{thm:nulldom} in the Supplement).
Thus, the PBJ provides approximate \emph{p}-values regardless of the error distribution.

In practice the joint distribution of the test statistics, $Q_0$, is not known {\it a priori} and must be estimated from the data.
In order to obtain asymptotically valid adjusted \emph{p}-values, the estimate for $Q_0$, $\hat Q_0$, must be consistent.
Because the distribution $Q_0 = Q_0(\Sigma)$ is a function of the covariance matrix $\Sigma$, a consistent estimator for $Q_0$ can be obtained from a consistent estimator for $\Sigma$.
The choice of the estimator is critical because $\hat \Sigma$ must be consistent for $\Sigma$ under the alternative distribution. That is, even if some null hypotheses are false $\hat \Sigma$ must be consistent.
The PBJ procedure uses a consistent estimator for $\Sigma$ based on the residuals of the full design $X$.
The consistency of $\hat \Sigma$ under the alternative guarantees asymptotic control of the FWER (see supplementary material).
Note that, in general, the PJ procedure may not yield a consistent estimator for the joint distribution $Q_0$ if the alternative is true at more than one location because the estimates of covariances are biased.
The covariance estimates are biased due to the fact that, for the reduced model used by the PJ MTP (see Procedure \ref{proc:permutation}), the mean is incorrectly specified in locations where the alternative is true. 
If the assumption of subset pivotality is satisfied, then the permutation estimator will be consistent.

\subsection{Parametric bootstrap null distribution}

For the parametric bootstrap we assume model \eqref{eq:model} and use F-statistics for the test $H_{0v}: \beta_v = 0$.

\begin{equation}\label{eq:Fstatistic}
F_{vn} = \frac{(n-m)Y_v^T(R_{X_0} - R_{X})Y_v}{m_1Y_v^TR_XY_v},
\end{equation}
where $R_A$ denotes the residual forming matrix for $A$.
When the errors are normally distributed \eqref{eq:normality}, $F_{vn}$ is an F-distributed random variable with $m_1$ and $n-m$ degrees of freedom.
When the errors are not normal the statistics \eqref{eq:Fstatistic} are asymptotically $m_1^{-1}\chi^2_{m_1}$.
The following theorem gives the asymptotic joint distribution of the statistics.

\begin{theorem}\label{thm:JointFdistribution}
Assume model \eqref{eq:model}, let $F_{vn}$ be as defined in \eqref{eq:Fstatistic}, and define the $p \times V$ matrix $\alpha = [ \alpha_1, \ldots, \alpha_V]$.
Further assume that, under the null, 
\begin{align}
R_{X_0} \E Y & = \E Y - X_0\alpha = 0 \label{eq:firstmoment}\\
\lVert \Psi \rVert_M & < \infty \label{eq:secondmoment},
\end{align}
where $\lVert \Psi \rVert_M = \sup_{x} \lVert \Psi x \rVert/\lVert x \rVert$ is the induced norm.

Then the following hold:
\begin{enumerate}
\item Define the matrix $\Phi_{i,i} = 1/\sqrt{\Psi_{i,i}}$ and $\Phi_{i,j} = 0$ for $i\ne j$. When \eqref{eq:normality} holds,
\begin{equation} \label{eq:wisharts}
\begin{array}{rl}
\Phi Y^T (R_{X_0} - R_{X}) Y \Phi \sim \mathcal{W}_V(m_1, \Sigma)\\
\Phi Y^T R_X Y \Phi  \sim \mathcal{W}_V(n-m, \Sigma)
\end{array}
\end{equation} where $\mathcal{W}_p(d, \Sigma)$ denotes a singular Wishart distribution with degrees of freedom $d<p$ and matrix $\Sigma \in \R^{p\times p}$ \citep{srivastava_singular_2003}.
\item The F-statistics converge in law to the diagonal of a singular Wishart distribution, that is,
\[
m_1 F_{n} = m_1[F_{1n}, \ldots, F_{Vn} ] \rightarrow_{L} \mathrm{diag}\left\{ \mathcal{W}_V(m_1, \Sigma) \right\}.
\]
\end{enumerate}
 as $n \to \infty$.
\end{theorem}
In order to make the statistics robust regardless of whether the errors are normal we use the transformation
\begin{equation}\label{eq:transformed}
Z_{vn} = \Phi^{-1}\{\Phi_n(F_{vn})\},
\end{equation}
where $\Phi_{n}$ is the cumulative distribution function (CDF) for a $\mathcal{F}(m_1, n-m)$ random variable and $\Phi^{-1}$ is the inverse CDF for a $\chi^2_{m_1}$ random variable.
If equation \eqref{eq:normality} is true, then the transformed statistics \eqref{eq:transformed} are marginally $\chi^2_m$ statistics.
When the errors are nonnormal then $Z_{vn}$ is asymptotically $\chi^2_{m_1}$ (see Theorem \ref{thm:nulldom} in the Supplement).
The asymptotic joint distribution of the statistics given in Theorem \ref{thm:JointFdistribution} allows us to use a diagonal singular Wishart to compute approximate adjusted \emph{p}-values.
To compute probabilities we do not have to sample the full matrix \eqref{eq:wisharts}, since only the diagonal elements, $\text{diag}\{Y^T (R_{X} - R_{X_F}) Y\}$, are required.

To find adjusted \emph{p}-values, $\tilde p_v$, for the single-step procedure we compute the probability
\begin{equation}
\label{eq:sspvalue}
\tilde p_v = \P\left( \max_{k\le V}\lvert Z_k \rvert > \lvert Z_{v0} \rvert \right),
\end{equation}
where
\begin{equation}
\label{eq:standardnormal}
Z = (Z_1, \ldots, Z_V) \sim \mathrm{diag}\left\{ \mathcal{W}_V(m_1, \Sigma) \right\},
\end{equation}
and $Z_{v0}$ is the observed statistic at location $v$.
Theorems \ref{thm:nulldom} and \ref{thm:aFWERc} in the Supplement guarantee asymptotic control of the FWER when $\eqref{eq:standardnormal}$ is used as the null distribution.

In practice the joint distribution \eqref{eq:standardnormal} is unknown due to the fact that $\Sigma$ is unobserved, so the probability \eqref{eq:sspvalue} cannot be computed.
We must obtain an estimate for $\Sigma$ in order to compute estimates of these probabilities.
Since the diagonal, $\Sigma_{v,v} = 1$, we only need to estimate the off-diagonal elements.
By estimating $\rho_{ij}$ with the consistent estimator
\begin{equation*}
\hat \rho_{jk} = \frac{Y_j^T R_X Y_k}{ \hat \sigma_j \hat \sigma_k},
\end{equation*}
we are guaranteed asymptotic control of the FWER (see the Supplement).
This estimator is biased toward zero in finite samples, and yields conservative estimates of the correlation.
Note that using the residuals of the full model is crucial here as that estimator yields consistent estimates of the correlation regardless of whether the alternative is true in each location.
Instead of using $\Sigma$ in \eqref{eq:standardnormal} we use the estimated covariance matrix of the test statistics
\begin{equation}\label{eq:sigmahat}
\hat\Sigma_{jk} =  \left\{
\begin{array}{lr}
1 & : j=k \\
\hat\rho_{jk} & : j \ne k
\end{array}\right..
\end{equation}
Importantly, the covariance of the tests statistics does not depend on what model parameter is being tested, so a single null distribution can be used for tests of all parameters provided the tests are on the same degrees of freedom. This conserves computing time relative to the permutation procedure which must estimate a null distribution for each test.

\subsection{The parametric-bootstrap procedures}

We compute \emph{p}-values using a parametric bootstrap: We use the estimate of $\hat \Sigma$ to generate $B$ diagonal singular Wishart statistics.
Because the rank of $\hat \Sigma$ is at most $\min\{(n-m), V\}$ it does not require the storage of the full $V\times V$ covariance matrix if $V > (n-m)$.
This gives the following procedure for estimating the null distribution using the parametric bootstrap.

\begin{procedure}[Parametric bootstrap null estimation]
Assuming the model \eqref{eq:model}:
\label{proc:sdpb}
 \begin{enumerate}
  \item Regress $Y_v$ onto $X$ to obtain the test statistics for $H_{v0}: \beta_v = \beta_{v0}$ using \eqref{eq:transformed}. Let $ Z_{(1)0} < Z_{(2)0} <, \ldots, < Z_{(V)0}$ denote the ascending test statistics and $\tilde E = [ \hat \epsilon_{(1),0}, \ldots, \hat \epsilon_{(V)0} ]$, the $n\times V$ matrix of their associated residuals from model \eqref{eq:model}.
  \item Standardize $\tilde E$ so that the column norms are 1. Denote the standardized matrix by $E$. Let $r = \text{rank}(E) = \min\{n-m, V\}$.
 \item Use $E$, $m_1$, and $r$ to generate the null test statistics $Z_{b}= (Z_{(1)0}, \ldots, Z_{(V)b})$ for $b=1,\ldots,B$.
 \begin{enumerate}
 \item Perform the singular value decomposition of $E = U D \tilde M^T$ where $D$ is an $r \times r$ diagonal matrix and $U$ and $M$ have orthonormal columns. Let $M = \tilde M D $.
 \item Generate an $ r \times m_1$ matrix $S_b$ with independent standard normal entries.
 \item Obtain the null statistics $Z_{b} = \text{diag}(M S_b S_b^T M^T)$. $Z_b$ are distributed according to a diagonal singular Wishart distribution.
 \end{enumerate}
 \end{enumerate}
\end{procedure}
The singular value decomposition only needs to be performed once for the entire procedure.
Computing the statistics does not require multiplying $M S_b S_b^T M^T$, because we only need the diagonal entries.
Procedures \ref{proc:bootSS} and \ref{proc:bootSD} are used to compute single-step and step-down adjusted \emph{p}-values from the bootstrap sample.
\section{Methods}
\label{sec:methods}

Code to perform the analyses presented in this manuscript is available at \url{https://bitbucket.org/simonvandekar/param-boot}. While the processed data are not publicly available, unprocessed data are available for download through the Database of Genotypes and Phenotypes (dbGaP) \citep{satterthwaite_philadelphia_2016}.
We provide simulated region-wise data with the code so that readers can perform the region-wise analyses presented here.

\subsection{Cerebral blood flow data description}
 
The Philadelphia Neurodevelopmental Cohort (PNC) is a large initiative to understand how brain maturation mediates cognitive development and vulnerability to psychiatric illness \citep{satterthwaite_neuroimaging_2014}.
In this study, we investigate image and regional measurements of CBF using an arterial spin labeling (ASL) sequence collected on 1,578 subjects, ages 8-21, from the PNC (see \citet{satterthwaite_neuroimaging_2014} for details).
Abnormalities in CBF have been documented in several psychiatric disorders including schizophrenia \citep{pinkham_resting_2011} and mood and anxiety disorders \citep{kaczkurkin_elevated_2016}.
Establishing normative trajectories in CBF is critical to characterizing neurodevelopmental psychopathology \citep{satterthwaite_impact_2014}.

T1-weighted structural images are processed using tools included in ANTs \citep{tustison_n4itk:_2010}.
Voxel-wise analyses and simulations are restricted to gray matter locations.
For region-wise analyses CBF is averaged within 112 anatomically defined gray matter regions.

The CBF image is co-registered to the T1 image using boundary-based registration \citep{greve_accurate_2009}, and normalized to the custom PNC adolescent template using the top-performing diffeomorphic SyN registration included in ANTs \citep{avants_reproducible_2011,klein_evaluation_2009}. Images were down-sampled to 2mm resolution and smoothed with a Gaussian kernel at a FWHM of 6mm prior to group-level analysis.

In order to parcellate the brain into anatomically defined regions, we use an advanced multi-atlas labelling approach which creates an anatomically labeled image in subject space \citep{avants_reproducible_2011}.
The ASL data are pre-processed using standard tools included with FSL \citep{jenkinson_improved_2002,satterthwaite_impact_2014}.
The T1 image is then co-registered to the CBF image using boundary-based registration \citep{greve_accurate_2009}, the transformation is applied to the label image, and then CBF values are averaged within each anatomically-defined parcel.

Of the 1,601 subjects with imaging data 23 did not have CBF data. Additionally, 332 were excluded due to clinical exclusionary criteria, which include a history of medical disorders that affect the brain, a history of inpatient psychiatric hospitalization, or current use of a psychotropic medication.
An additional 274 subjects were excluded by an imaging data quality assurance procedure which included automatic and manual assessments of data quality and removal of subjects with negative mean CBF values in any of the anatomical parcels.
These exclusions yielded a total of 972 subjects used for the imaging simulations and analysis.

\subsection{Simulations}

\subsubsection{Synthetic simulations}

In order to better understand convergence rates for joint MTP procedures we perform simulations where the data generating covariance structure is known.
The parametric procedures like Holm and PBJ rely on approximations due to estimating the covariance structure and using a multivariate normal approximation.
We assume a normal data generating distribution so that the test statistics are T-distributed.
Specifically, we assume two samples
\begin{align*}
X_i & \sim \mathcal{N}(\mu_0, \Sigma) \text{ for $i \le n/2$} \\
X_i & \sim \mathcal{N}(\mu_1, \Sigma) \text{ for $i > n/2$},
\end{align*}
where $\mu_k \in \R^V$, and we perform the tests $H_{0v}: \hat \mu_{1v} - \hat \mu_{0v}=0$ for $v=1,\ldots,V$.
We vary the values $n\in \{40, 80, 100, 200\}$, $p\in \{100, 200, 500, 1000, 10000\}$, $\Sigma_{jj} = 1$, and $\Sigma_{jk} \in \{0, 0.9^{\lvert j - k\rvert}, (-0.9)^{\lvert j - k\rvert} \} $ for $j\ne k$.
We make the first 10\% of the components of $\mu_1$ nonzero with parameter $\mu_{1v} = 0.4$ and all other mean parameters zero. We perform 500 simulations to estimate FWER and power.

Several step-down MTPs are considered.
For insight into convergence properties, we let some MTPs rely on the unobserved covariance matrix, so they are not possible in practice.
We use the following notation:
$(T_n \mid \text{Holm})$ denotes Holm's procedure where a standard normal distribution is used to compute \emph{p}-values for the T-statistics.
$(Z_n \mid \text{Holm})$ is Holm's procedure where the T-statistics are transformed using \eqref{eq:transformed}, that is, the T-distribution is used to compute \emph{p}-values. $(T_n \mid \text{Holm})$ and $(Z_n \mid \text{Holm})$ both demonstrate the how conservative using Holm procedure is and compare the effect of a normal approximation to the marginal densities.
$(T_n \mid \Sigma)$ denotes adjusted \emph{p}-values computed using PBJ with Procedure \ref{proc:bootSD} using the true covariance $\Sigma$ and the raw T-statistics.
The FWER for $(T_n \mid \Sigma)$ gives us an idea of the sample size required for approximating a multivariate T-statistic with a normal distribution.
$(T_n \mid \hat \Sigma)$ uses PBJ with the sample estimate $\hat \Sigma$ using the raw T-statistics.
$(Z_n \mid \hat \Sigma)$ uses PBJ with the sample estimate $\hat \Sigma$ and the transformation \eqref{eq:transformed}.
$(T_n \mid \text{Perm})$, uses permutations to generate the joint distribution and untransformed T-statistics.
Within each simulation 1000 bootstraps and permutations are used to compute \emph{p}-values for the PBJ and PJ MTPs.

\subsubsection{Real data simulations}
To create realistic simulated data sets we use samples generated from real data to compare the FWER and power of the MTPs for region-wise ($V=112$) and voxel-wise analyses ($V=$127,756).
For each simulation we draw subsamples without replacement from the CBF data for sample sizes of $n=40,100,200,400$.
The region-wise simulations cover the case where $p>n$ and $p \approx n$.
For the voxel-wise simulations we smooth at FWHM$=6$.
We present voxel-wise results with different smoothing kernels in the Supplement.

Results from the the synthetic simulations demonstrate that deviations from normality increases the FWER above the nominal level (see Section \ref{sec:synthsim}).
For this reason, we perform the Yeo-Johnson transformation prior to performing hypothesis tests in the real data simulation and CBF analyses \citep{yeo_new_2000}.
The Yeo-Johnson transformation is a single parameter transformation similar to the Box-Cox that allows negative values for the outcome variable.
We estimate the parameter at each location in the image using a profile-likelihood approach \citep{yeo_new_2000}.
Inference is performed conditional on the estimated parameter.

In each simulation we fit the following model with real covariates including age, sex, race and motion (mean relative displacement; MRD) as well as artificially generated covariates
\begin{equation*}
\E Y_{iv} = \alpha_0 + \alpha_1\text{age}_i + \alpha_2\text{sex}_i + \alpha_3\text{race}_{i1} + \alpha_4\text{race}_{i2} + \alpha_5\text{MRD}_{i} + \sum_{j=1}^3 \beta_{jv} g_j
\end{equation*}
where $g_j$ are indicators for an artificial factor with 4 levels to represent different clinical groups in equal proportions and $\E Y_{iv}$ represents the conditional expectation of the transformed outcome.
Multiple groups were generated so that we could perform a test of the parameter for the second group indicator
\begin{equation}\label{eq:onedof}
H_{0v} : \beta_{1v}=0
\end{equation}
on one degree of freedom, and the test of 
\begin{equation}\label{eq:threedof}
H_{0v} : \beta_{jv} = 0 \text{ for all $j$}
\end{equation}
on 3 degrees of freedom.
To assess power, in each simulation we generate signal in randomly selected locations.
For the region-wise simulation we randomly selected 3 brain regions $v_k \in \{ 1, \ldots, 112\}$ for $k=1,2,3$, by setting $\beta_{1v_k}=10$ and $\beta_{jv} = 0$ otherwise.
For the voxel-wise simulations we first select a random gray matter voxel $v_0$ and create a cube with a radius of 6 voxels centered at $v_0$.
Let $N_{v_0} \subset\{1, \ldots, 127756\}$ denote the gray matter voxels within the cube.
We create a parameter image where $\beta_{1v}=120$ for all $v \in N_{v_0}$ and $\beta_{1v} = 0$ otherwise.
The generated parameter image is smoothed at FWHM$=6$mm and is added to the CBF images smoothed with the same kernel.
All other artificial parameters were set to $0$.

We use 1000 simulations to estimate FWER and power for the region-wise data and 500 simulations for the voxel-wise data, which take considerably longer to run.
All results are presented for the rejection level $\alpha=0.05$.
FWER is defined as the proportion of simulations where any true null hypothesis was rejected and power is estimated by the mean number of rejected hypotheses amongst the false null hypotheses.
For the voxel-wise simulations the false null hypotheses were defined as all voxels in the smoothed parameter image where $\beta_{v1}>1$.
We use Wilson confidence intervals for the FWER estimated from the 500 simulations implemented within the {\tt binom} package in {\tt R} \citep{wilson_probable_1927,dorai-raj_binom_2014,brown_interval_2001}.
For power results we use a normal approximation for confidence intervals. Note however, that the variance estimate for the power results will be biased downward since the generated subsamples are dependent.

We compare the FWER and power of the Bonferroni, Holm, PBJ, and PJ MTPs.
Permutation tests are performed using Randomise \citep{winkler_permutation_2014}, which supports a wide array of possible tests, but only performs the single-step method.
In the region-wise simulations we assess the single-step and step-down PBJ MTP.
For the voxel-wise data we only compare single-step methods for the PBJ and PJ procedures because accurately estimating $V$ cutoffs when $V$ is large requires an infeasibly large number of samples.
For the region-wise simulations we use 5000 samples for both the PBJ and PJ MTPs.
For the voxel-wise simulations we use 5000 bootstraps for the PBJ, and 500 permutations with Randomise.
We use 500 permutations with Randomise because the procedure takes considerable time to run.
Mean run times to perform 5000 simulations are given in Figure \ref{fig:comptime}.
Estimates of the \emph{p}-values are unbiased for both procedures, but the \emph{p}-values for the permutation procedure will have higher variance as fewer permutations are used (see \cite{winkler_faster_2016}).
Wilson confidence intervals for a true \emph{p}-value of $p=0.05$ for 500 and 5000 permutations are $[0.03, 0.07]$ and $[0.04, 0.06]$.

In Section \ref{sec:Simulations} we discuss the computational complexity of the PJ procedure versus the parametric-bootstrap.
To compare actual computing time we take the mean of the time to perform the 1 and 3 degree of freedom tests with 5000 simulations for each method for the region-wise analyses.
For the computing times of the voxel-wise analyses we multiply the computing time for the PJ procedure by 10 because 10 times fewer permutations were used than bootstraps.
Note that this slightly over estimates the PJ computing time as it also multiplies the image load time, which is 1-2 minutes depending on the sample size.

\subsection{Cerebral blood flow statistical analysis}
\label{sec:CBFmethods}
We perform region- and voxel-wise analyses of CBF in order to identify locations where there are sex differences in development-related changes of CBF.
For the region-wise analysis, we test the sex by age interaction on the average CBF trajectories in the $V=112$ regions using an F-statistic from an unpenalized spline model.
For the voxel-wise analysis, we perform the same test in all $V=$127,756 gray matter voxels.

For the region-wise data we fit the age terms with thin plate splines with 10 degrees of freedom.
Thus, the numerator of the F-statistic has 9 degrees of freedom using the {\tt mgcv} package in {\tt R} \citep{R_2016,wood_fast_2011}.
Results from the simulation analyses and previous theoretical results \citep{gotze_rate_1991} demonstrate that multivariate convergence rates depends on the dimension of the vector of statistics (Tables \ref{tab:n40error} and \ref{tab:n100error}), and the degrees of freedom of the test (Figure \ref{fig:imagingt-statistics}).
For this reason, we use the Yeo-Johnson transformation for the PBJ procedure so that the transformed CBF data are approximately normal \citep{yeo_new_2000}.
We estimate the age terms for the voxelwise data on 5 degrees of freedom, so that the test for the interaction is on 4 degrees of freedom.
Race and motion (mean relative displacement; MRD) are included as covariates at each location for both analyses.
Similar results were presented in a previous report \citep{satterthwaite_impact_2014} using Bonferroni adjustment.
We also perform the voxel-wise analyses using a 10 degrees of freedom spline basis in the Supplement.

The same MTPs are compared for the region- and voxel-wise CBF analysis as used in simulations.
For the region- and voxel-wise analyses 10,000 samples are used for the PBJ and PJ procedures.
We present the corrected results with the FWER controlled at $\alpha=0.01$ for the region-wise analysis and $\alpha=0.05$ for the voxel-wise analysis.
A more conservative threshold is used for the region-wise data as the smaller number of comparisons and noise reduction from averaging within regions increases power considerably.

\section{Simulation Results}
\label{sec:Simulations}

\subsection{Synthetic simulation results}
\label{sec:synthsim}

We use simulations to explore how the FWER is affected by using asymptotic approximations and estimating the covariance matrix in finite sample sizes.
Results are shown for sample sizes of $n=40$ and $n=100$ (Tables \ref{tab:n40error} and \ref{tab:n100error}).
The results demonstrate that by using the sample covariance estimate in the PBJ MTP the FWER is only slightly inflated for $n=40$ and when $n=100$ the FWER is controlled at the nominal level for all the dimensions considered ( Column $Z_n \mid \hat\Sigma$).
Note, that when the transformation \eqref{eq:transformed} is not used all procedures have inflated FWERs (Columns $T_n$ in Tables \ref{tab:n40error} and \ref{tab:n100error}) due to the fact that the multivariate normal approximation is not accurate and is worse for a larger number of tests \citep{gotze_rate_1991}.
We can conclude that most of the error in estimating the FWER comes from the normal approximation.
Even when $\hat \Sigma$ is rank deficient it still provides nominal FWER control.
Interestingly, the PJ procedure controls the FWER for all the sample sizes considered.
While permutation tests are exact for univariate distributions \citep{lehmann_testing_2006}, to our knowledge there is no theoretical justification that multivariate permutations are accurate when the number of statistics exceeds the sample size.
Finally, the PBJ and PJ MTPs control the FWER at the nominal level, while Holm's procedure is conservative for correlated test statistics.

\begin{table}[p!]
\centering
\begin{tabular}{llrrrrrr}
  \hline
& $n=40$ &  $T_n \mid \text{Holm}$ & $Z_n \mid \text{Holm}$ & $T_n \mid \Sigma$ & $T_n \mid \hat \Sigma$ & $Z_n \mid \hat \Sigma$ & $T_n \mid \text{Perm}$ \\ 
  \hline
  \multirow{6}{*}{\rotatebox[origin=c]{90}{Indep}}
&$m =  100 $ & 12 & 6 & 12 & 16 & 6 & 5 \\ 
 & $m =  200 $ & 13 & 4 & 13 & 17 & 4 & 4 \\ 
  &$m =  500 $ & 16 & 6 & 17 & 21 & 7 & 5 \\ 
  &$m =  1000 $ & 20 & 4 & 20 & 23 & 5 & 3 \\ 
  &$m =  5000 $ & 27 & 5 & 28 & 29 & 6 & 3 \\ 
  &$m =  10000 $ & 38 & 5 & 39 & 40 & 9 & 5 \\
 \hline
\multirow{6}{*}{\rotatebox[origin=c]{90}{Pos AR(1)}} 
  &$m =  100 $ & 6 & 2 & 10 & 12 & 5 & 5 \\ 
  &$m =  200 $ & 8 & 3 & 13 & 13 & 5 & 6 \\ 
  &$m =  500 $ & 9 & 3 & 14 & 18 & 6 & 4 \\ 
  &$m =  1000 $ & 13 & 2 & 17 & 22 & 5 & 5 \\ 
  &$m =  5000 $ & 22 & 2 & 27 & 32 & 4 & 3 \\ 
  &$m =  10000 $ & 26 & 4 & 33 & 34 & 7 & 6 \\ 
 \hline
  \multirow{6}{*}{\rotatebox[origin=c]{90}{Neg AR(1)}} 
  &$m =  100 $ & 7 & 2 & 11 & 12 & 6 & 5 \\ 
  &$m =  200 $ & 6 & 3 & 11 & 16 & 5 & 4 \\ 
  &$m =  500 $ & 7 & 2 & 11 & 18 & 4 & 3 \\ 
  &$m =  1000 $ & 12 & 3 & 17 & 19 & 6 & 4 \\ 
  &$m =  5000 $ & 20 & 3 & 25 & 29 & 7 & 4 \\ 
  &$m =  10000 $ & 27 & 4 & 34 & 35 & 6 & 5 \\ 
   \hline
\end{tabular}
\caption{Type 1 error results for $n=40$ to assess convergence rates. Values are mean percentage of correctly rejected tests across 500 simulations. Test statistics were simulated as normal with independent (Indep), positive autoregressive (Pos AR(1)), and negative autoregressive (Neg AR(1)) correlation structures with $\rho=0.9$ and $\rho=-0.9$. The number of tests was varied within $m=(200, 500, \text{1,000}, \text{5,000}, \text{10,000})$ with 10\% non-null test statistics. Detailed descriptions of the column names are given in Section \ref{sec:methods}.
Results demonstrate that the majority of error in the PBJ procedure is due to the convergence of the T-statistics to a normal distribution.}
\label{tab:n40error}
\end{table}

\begin{table}[p!]
\centering
\begin{tabular}{llrrrrrr}
  \hline
& $n=100$ &  $T_n \mid \text{Holm}$ & $Z_n \mid \text{Holm}$ & $T_n \mid \Sigma$ & $T_n \mid \hat \Sigma$ & $Z_n \mid \hat \Sigma$ & $T_n \mid \text{Perm}$ \\ 
  \hline
  \multirow{6}{*}{\rotatebox[origin=c]{90}{Indep}} 
&$m =  100 $ & 6 & 6 & 6 & 8 & 6 & 4 \\ 
  &$m =  200 $ & 6 & 4 & 7 & 8 & 5 & 5 \\ 
  &$m =  500 $ & 7 & 6 & 8 & 8 & 6 & 5 \\ 
  &$m =  1000 $ & 10 & 4 & 10 & 10 & 5 & 6 \\ 
  &$m =  5000 $ & 10 & 4 & 11 & 10 & 4 & 5 \\ 
  &$m =  10000 $ & 14 & 4 & 15 & 14 & 4 & 5 \\ 
  \hline
  \multirow{6}{*}{\rotatebox[origin=c]{90}{Pos AR(1)}} 
  &$m =  100 $ & 2 & 1 & 6 & 4 & 4 & 3 \\ 
  &$m =  200 $ & 4 & 3 & 6 & 8 & 5 & 4 \\ 
  &$m =  500 $ & 5 & 3 & 7 & 8 & 4 & 5 \\ 
  &$m =  1000 $ & 7 & 2 & 8 & 9 & 5 & 6 \\ 
  &$m =  5000 $ & 6 & 3 & 9 & 8 & 4 & 4 \\ 
  &$m =  10000 $ & 9 & 3 & 11 & 10 & 4 & 3 \\ 
  \hline
  \multirow{6}{*}{\rotatebox[origin=c]{90}{Neg AR(1)}} 
  &$m =  100 $ & 5 & 3 & 7 & 9 & 6 & 5 \\ 
  &$m =  200 $ & 3 & 2 & 5 & 6 & 4 & 3 \\ 
  &$m =  500 $ & 4 & 2 & 8 & 8 & 5 & 4 \\ 
  &$m =  1000 $ & 5 & 3 & 8 & 8 & 4 & 4 \\ 
  &$m =  5000 $ & 8 & 4 & 10 & 9 & 6 & 5 \\ 
  &$m =  10000 $ & 8 & 3 & 11 & 9 & 6 & 5 \\ 
   \hline
\end{tabular}
\caption{Type 1 error results for $n=100$ to assess convergence rates. See Table \ref{tab:n40error} for details.}
\label{tab:n100error}
\end{table}

\subsection{Region-wise FWER and power}
We use simulations that sample from real imaging data to assess the FWER in finite samples for region-wise analyses.
For the test of hypothesis \eqref{eq:onedof}, the PBJ and PJ procedures maintain the nominal FWER for all samples (Figure \ref{fig:t-statistics}).
As expected, the FWER for the Bonferroni and Holm procedure are conservative.
The power of the joint testing procedures is higher than the marginal testing procedures (Figure \ref{fig:t-statistics}).
The PBJ has equal power to the PJ procedure.
The step-down procedure confers almost no benefit.

\begin{knitrout}
\definecolor{shadecolor}{rgb}{0.969, 0.969, 0.969}\color{fgcolor}\begin{figure}[p!]

{\centering \includegraphics[width=\maxwidth]{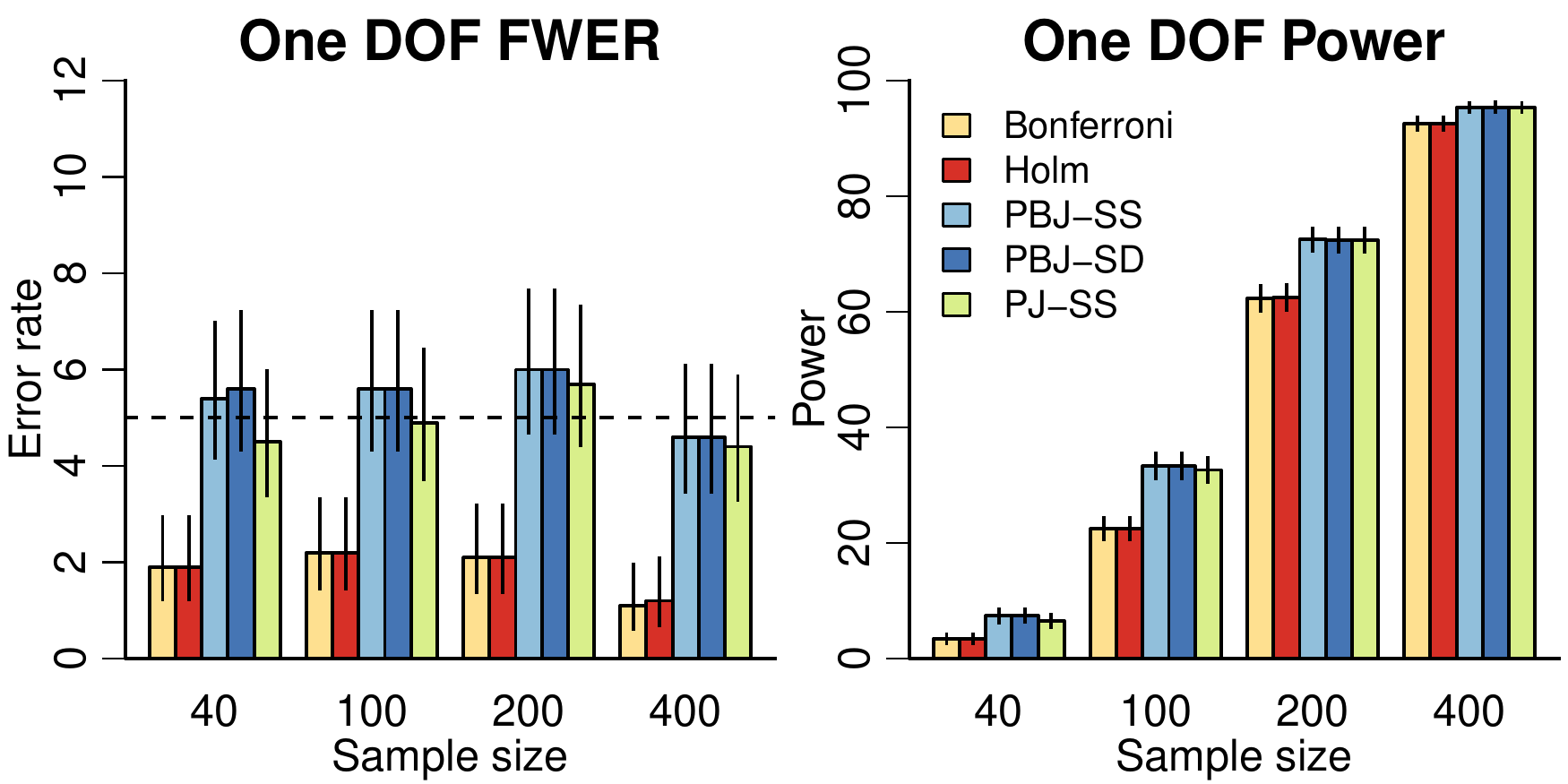} 
\includegraphics[width=\maxwidth]{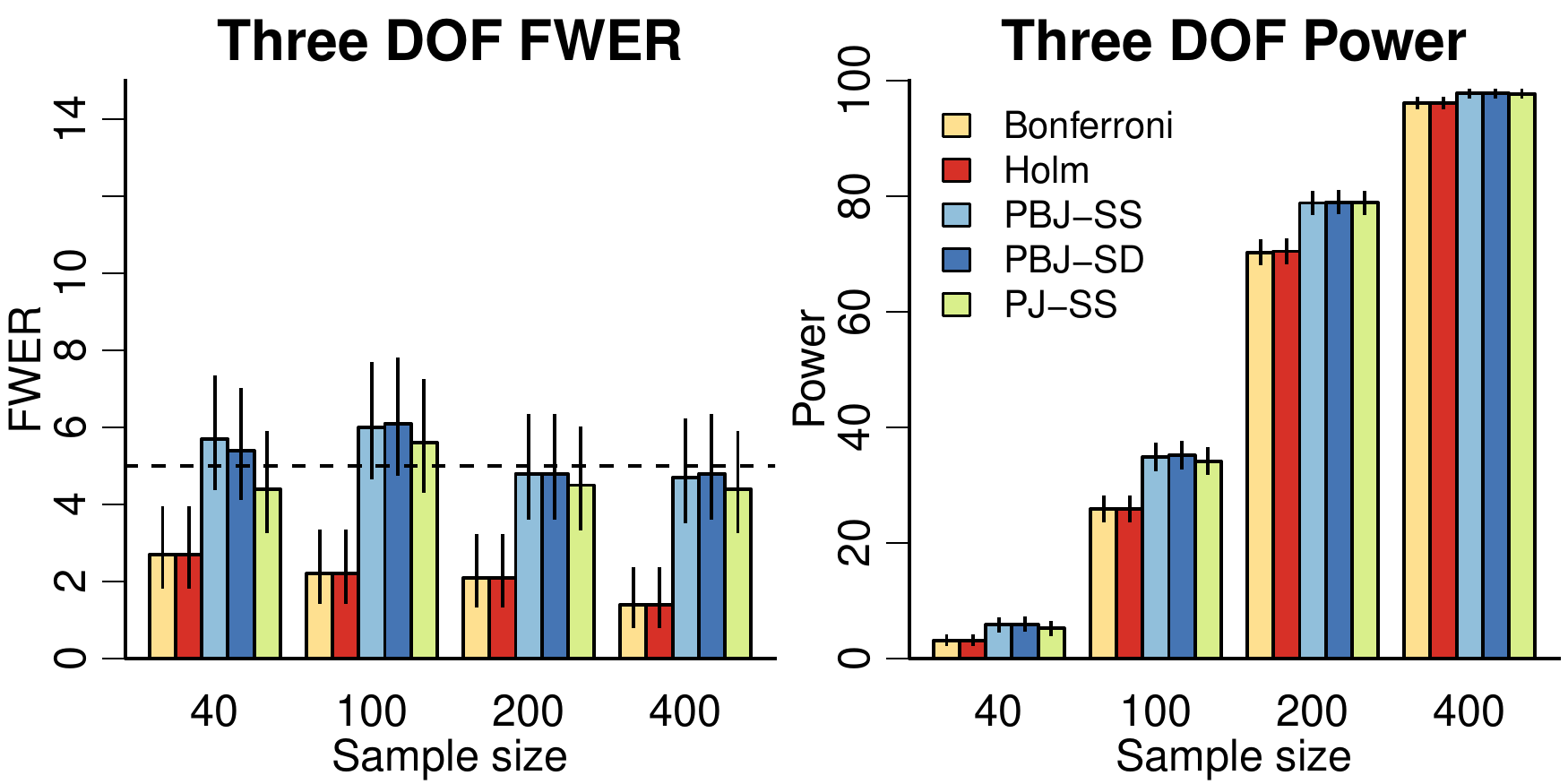} 

}

\caption{FWER and power for the F-statistic of \eqref{eq:onedof} on one degree of freedom (DOF)  and F-statistic of \eqref{eq:threedof} on three DOF for $V=112$ brain regions. Bonferroni and Holm both have conservative control. The PBJ maintains accurate FWE control even when $n<p$. The power for the PBJ and PJ procedures are equal. SS$=$Single-step; SD$=$Step-down. Lines indicate 95\% confidence intervals.}\label{fig:t-statistics}
\end{figure}

\end{knitrout}
For testing the hypothesis \eqref{eq:threedof} the PBJ and PJ procedures control the FWER at the nominal level for all sample sizes (Figure \ref{fig:t-statistics}).
The Bonferroni and Holm procedures give similarly conservative FWERs as the single degree of freedom test.
Power analyses demonstrate that the PBJ and PJ procedures have the same power.
As with testing \eqref{eq:onedof} the step-down procedure shows little improvement over the single-step.

\subsection{Voxel-wise FWER and power}
We use simulations that sample from real imaging data to assess the type 1 error and power for voxel-wise analyses.
For the test of \eqref{eq:onedof}, the PBJ maintains the nominal FWER for samples sizes greater than 200.
The PJ procedure maintains control for all sample sizes considered (Figure \ref{fig:imagingt-statistics}).
As expected, Bonferroni and Holm procedures have conservative FWER.
This conservative FWER leads to a reduction in power for these methods (Figure \ref{fig:imagingt-statistics}).
The power for PBJ was comparable to the PJ procedure.

\begin{knitrout}
\definecolor{shadecolor}{rgb}{0.969, 0.969, 0.969}\color{fgcolor}\begin{figure}[p!]

{\centering \includegraphics[width=\maxwidth]{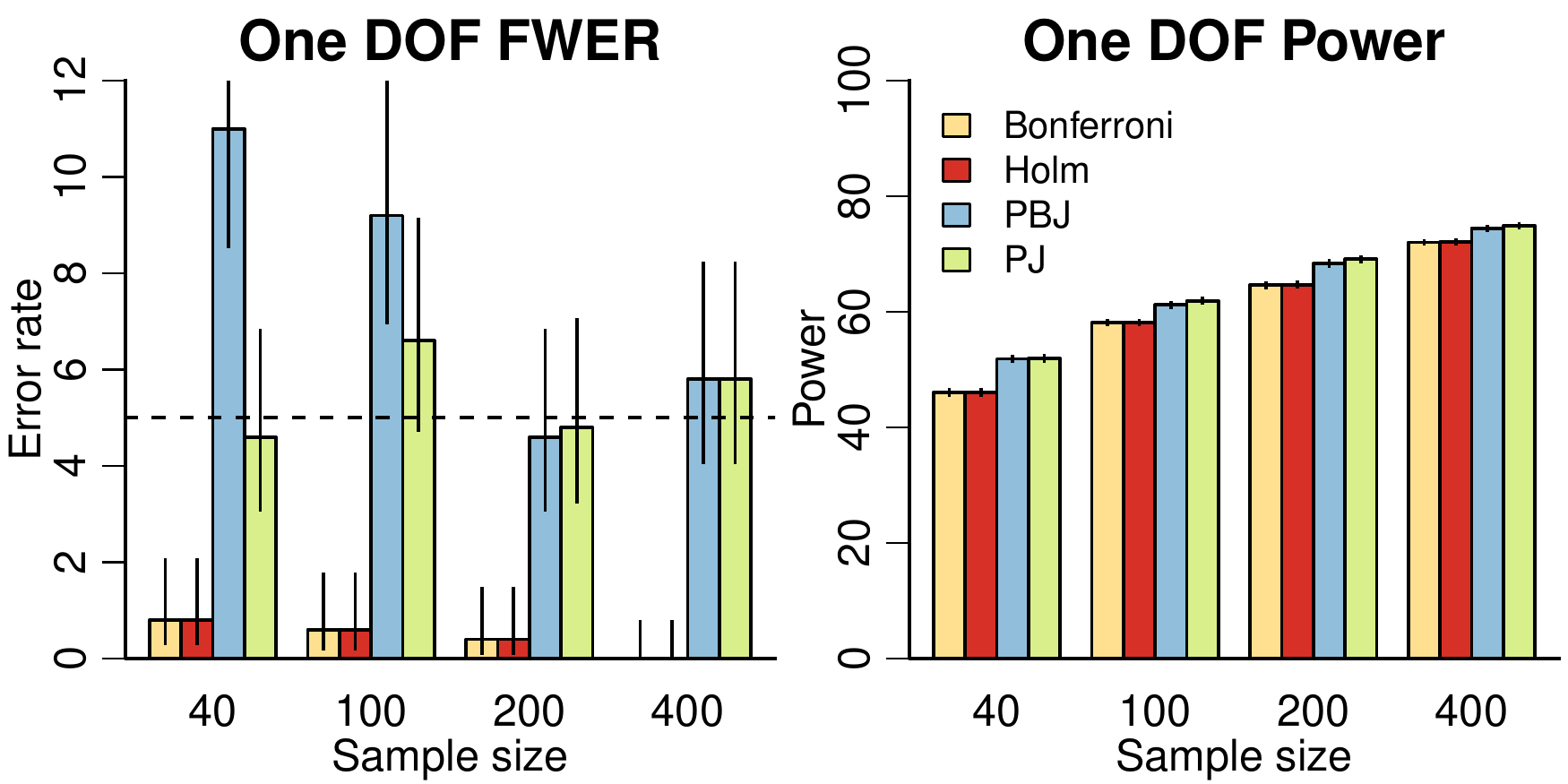} 
\includegraphics[width=\maxwidth]{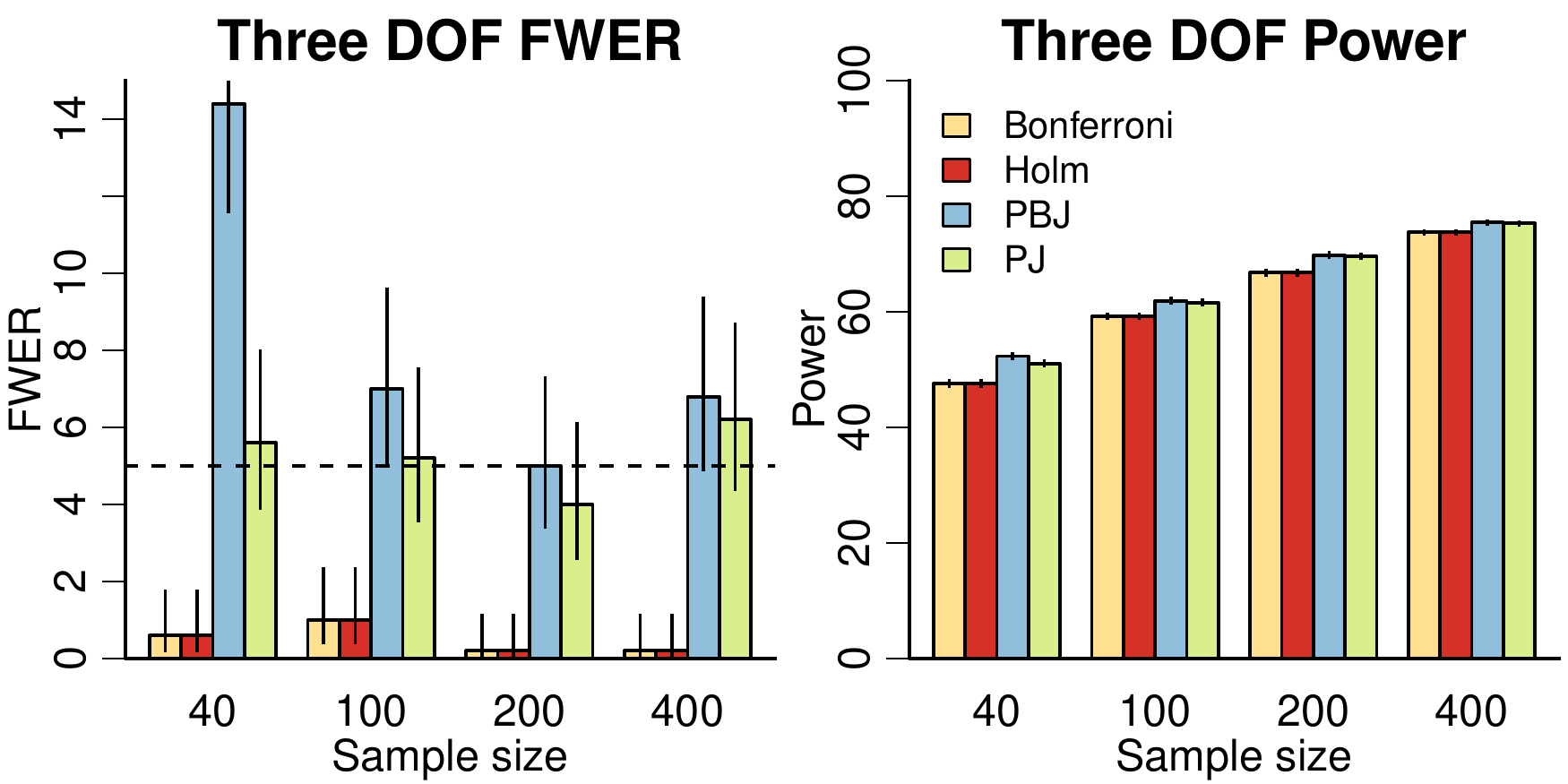} 

}

\caption{FWER and power for the F-statistic of \eqref{eq:onedof} on one degree of freedom (DOF) and F-statistic of \eqref{eq:threedof} on three DOF for the CBF image with $V=$127,756 voxels. Bonferroni and Holm both have conservative control.  The PBJ controls the FWER for samples sizes greater than 200. The power for the PBJ is approximately equal to the PJ tests. Lines indicate 95\% confidence intervals.}\label{fig:imagingt-statistics}
\end{figure}

\end{knitrout}

 For the test of \eqref{eq:threedof} on 3 degrees of freedom the PBJ controls the FWER for sample sizes larger than 200 (Figure \ref{fig:imagingt-statistics}).
 The PBJ procedure does not control the FWER for smaller sample sizes because the estimate of the covariance structure \eqref{eq:sigma} is not accurate enough to work as a plug-in estimator for the full rank covariance matrix.
 The PJ procedure maintains nominal control of the FWER for all sample sizes.
 The marginal testing procedures have the same conservative FWER control as above.
 Both joint testing procedures have greater power than the marginal procedures.

\subsection{Computing time}

While the PJ and PBJ procedures are both linear in the sample size, as the PBJ procedure works directly with the distribution of the test statistics we expect it to be faster than permutations.
To compare observed computing times we took the ratio of the time to evaluate both tests for PJ and PBJ procedures. The computing times for the sample sizes simulated for the region-wise data are given in Figure \ref{fig:comptime}.
The PBJ is at least twice as fast as the PJ procedure and increases with the sample size so that the computing time reduction is larger for larger sample sizes.

\begin{knitrout}
\definecolor{shadecolor}{rgb}{0.969, 0.969, 0.969}\color{fgcolor}\begin{figure}[p!]

{\centering \includegraphics[width=\maxwidth]{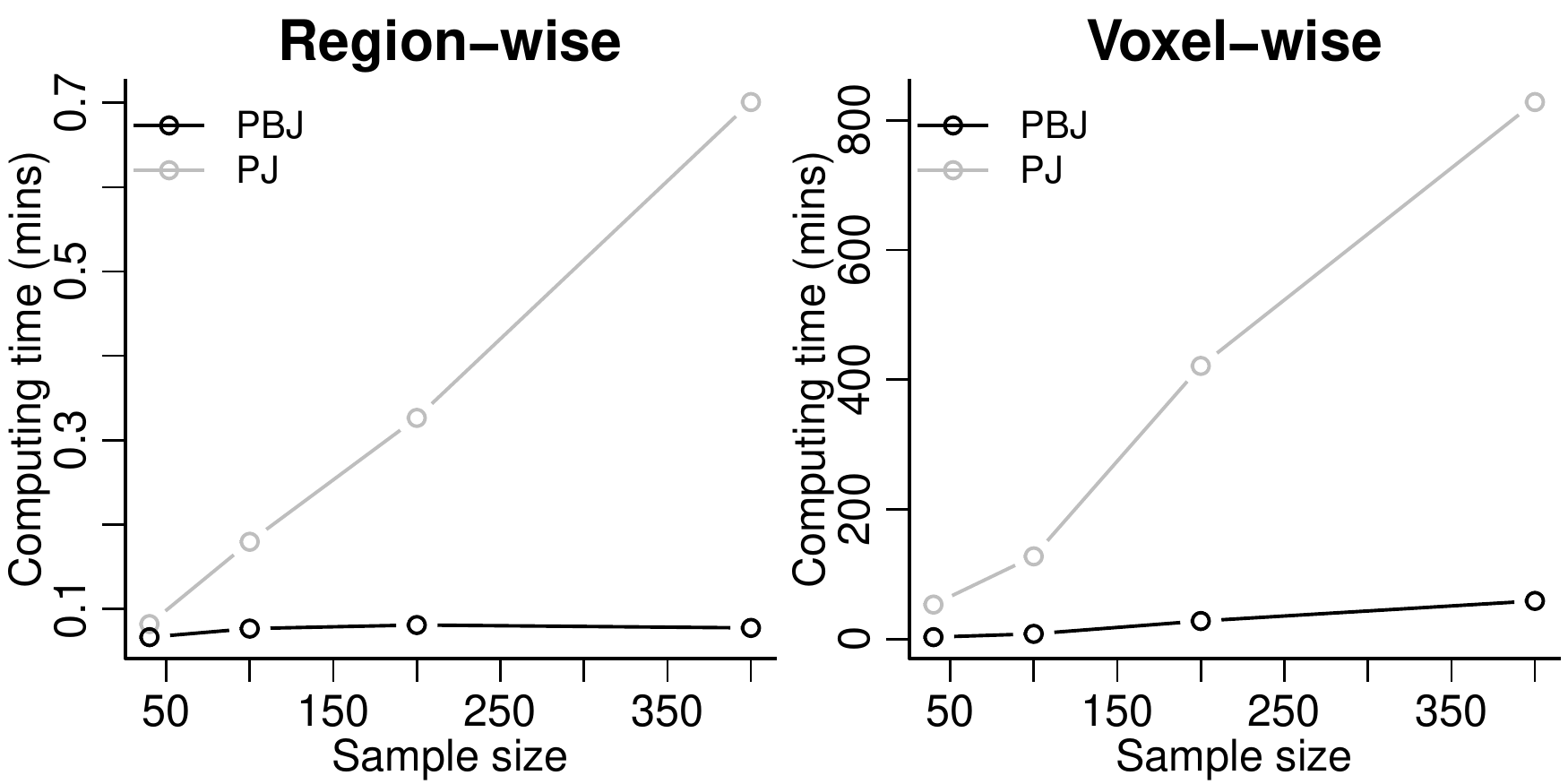} 

}

\caption[Computing times by sample size for the PBJ and PJ testing procedures for the region- and voxel-wise simulation analyses]{Computing times by sample size for the PBJ and PJ testing procedures for the region- and voxel-wise simulation analyses. We multiply the PJ computing time by 10 for the voxel-wise times because 10 times fewer permutations were used for that procedure.}\label{fig:comptime}
\end{figure}

\end{knitrout}

\section{Cerebral Blood Flow Results}
\label{sec:results}

To compare the testing procedures in the CBF data we perform a test for a nonlinear age-by-sex interaction on CBF trajectories for region- and voxel-wise analyses. For the voxel-wise analysis the PBJ and PJ MTPs take 5.9 and 69.2 hours to run, respectively. 
We perform an F-test for the interaction at each of the 112 regions (Figure \ref{fig:pncresults}).
The Bonferroni, Holm, PBJ, and PJ procedures reject $56$, $71$, $78$, and $61$ regions, respectively.
The PJ procedure is more conservative for two reasons. The first is that it is a single-step procedure; when there is a relatively large number of rejected tests then using a step-down procedure is more likely to improve power. The second reason is that the finite sample distribution of each of the regions is different. Regions near the edge of the brain are likely to be more heavily skewed due to imperfections in the image registration. By comparing all regions to the distribution of the maximum the PJ procedure is necessarily conservative because it compares to the most heavily skewed regions.
In contrast, by transforming the data prior to using the PBJ procedure the marginal distribution of the test statistics are approximately equal across regions.

\begin{figure}
\centering\includegraphics[width=\maxwidth]{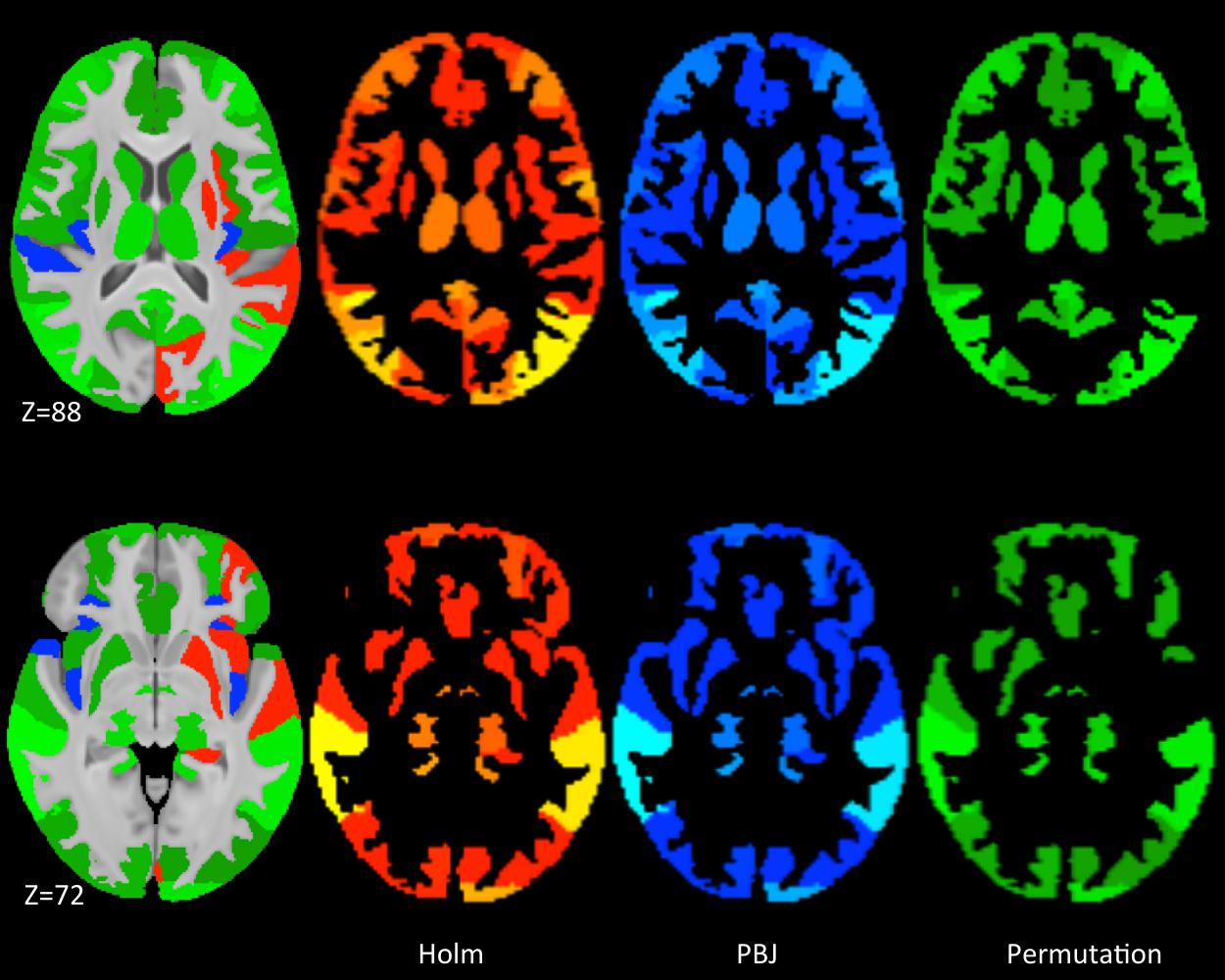}
\caption{FWER controlled results at $\alpha=0.01$ for Holm (red), PBJ step-down (blue), and PJ single-step (green) for the region-wise analysis. Color scale is $-\log_{10}(p)$ and shows results greater than 2. The left-most images show the overlay of PBJ, Holm, and PJ in that order. The color images show regions identified by Holm, PBJ, and NPBJ.}\label{fig:pncresults}
\end{figure}

For the voxel-wise analysis we perform the F-test for the nonlinear interaction on 4 degrees of freedom.
The single-step PBJ offers improved power over the Holm and PJ procedures (Figure \ref{fig:voxelpncresults}).
The Holm procedure ignores the covariance structure of the test statistics so yields conservative results.
The PJ procedure is more conservative even than the Holm procedure.
As with the region-wise analysis this is likely because the finite sample distribution of the test statistics is different: voxels near the edge of the brain tend to have higher variance and are likely heavily skewed.
If there is a subset of voxels with a heavily skewed distribution then taking the maximum test statistic will yield conservative inference for all locations that have tighter distribution.
By transforming the distribution of the voxels to be approximately normal the PBJ procedure offers improved power and speed.

\begin{figure}
\centering\includegraphics[width=\maxwidth]{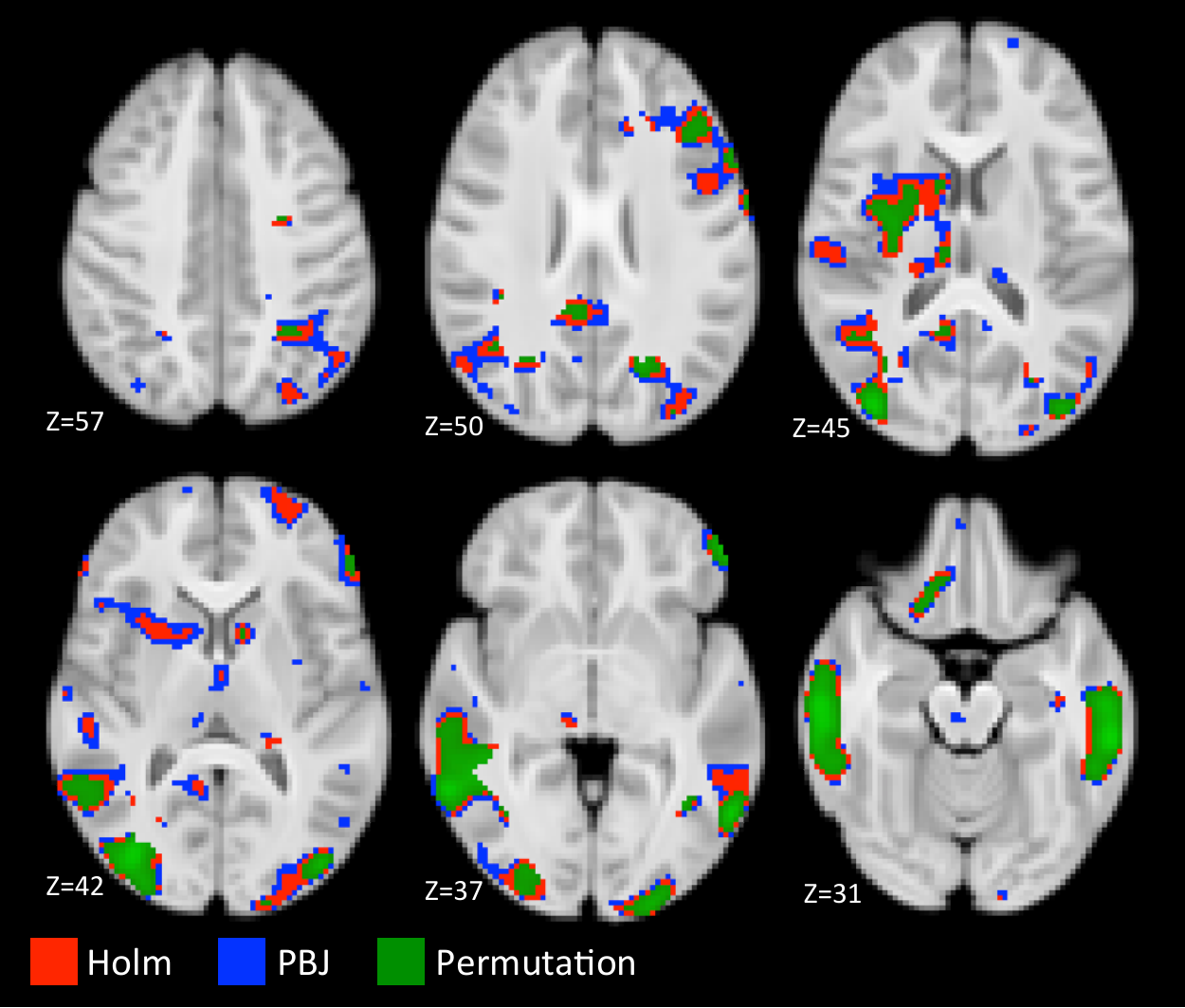} 
\caption{FWER controlled results at $\alpha=0.05$ for Holm (red), PBJ single-step (blue), and PJ single-step (green) for the voxel-wise analysis. Color scale is $-\log_{10}(p)$ for the adjusted \emph{p}-values and shows results greater than 1.3. The overlay order is the PBJ, Holm, and PJ procedures, so that green indicates regions where all three regions reject the null, red and green indicate regions where Holm and PJ reject, and the union of all colors is where PBJ rejects. Blue indicates locations where only the PBJ procedure rejects.}\label{fig:voxelpncresults}
\end{figure}

\section{Discussion}

We introduced a fast parametric bootstrap joint testing procedure as a new tool for multiple comparisons in neuroimaging.
The PBJ procedure improves computing time by generating the test statistics directly instead of permuting the original data.
If normality assumptions about the data generating distribution do not hold, then the Yeo-Johnson transformation can be used to obtain statistics that are approximately normal to improve the finite sample performance of the procedure.

In the CBF data analysis the PBJ is more powerful than the PJ MTP because the PJ MTP does not account for the fact that the finite sample distribution of the test statistics can be different.
Differences in the finite sample distribution of the statistics are attributable to certain regions near the edge of the brain having larger variance and skew.
For this reason taking the maximum across locations leads to conservative inference in locations that actually have tighter tails.
While the PBJ generates from a chi-squared distribution this ensures that a few heavy-tailed locations do not affect the distribution of the maximum.

In simulations, the step-down procedures provide little improvement in power over the single-step procedures.
However, in the regionwise analyses Holm rejected $15$ more regions than the Bonferroni procedure.
The reason for the difference is that step-down procedures offer little benefit when there is a small number of false null hypotheses and a large number of tests.

Using simulations we found that both joint procedures perform well, in some cases even when the number of tests exceeds the sample size.
This is quite surprising as it seems impossible that any given estimate of the joint distribution will satisfactorily reproduce the true joint distribution of the test statistics.
For example, if we consider the case of normal test statisics, $Z_n = (Z_{1n}, \ldots, Z_{Vn}) \sim N(0, \Sigma)$ with full rank covariance, then the sufficient statistic is $\hat \Sigma_n$, which can be of rank $n$ at most.
So, the probabilities generated conditioning on $\hat \Sigma_n$ assume $Z_n$ is restricted to a linear subspace of $\R^V$.
With nonnormal test statistics and more complex dependence structures it can only be more difficult to reproduce the null distribution.

MRI is a flexible noninvasive tool for studying neural aberrations related to psychiatric disorders such as schizophrenia \citep{pinkham_resting_2011} and mood and anxiety disorders \citep{kaczkurkin_elevated_2016}.
However, recent studies have shown that the PJ MTP is the only inference methodology to reliably control the FWER in neuroimaging data \citep{eklund_cluster_2016,silver_false_2011}.
We have shown that inference using the currently available permutation procedure can take days and lead to conservative inference.
Our proposed PBJ MTP is a reliable and fast testing procedure that will be a critical tool in studying functional and physiological features that can improve our understanding of the brain and its relation to behavior.


 \section*{Acknowledgments}
 We thank Mark van der Laan for a helpful discussion regarding null estimation.
 The PNC was funded by RC2 grants from the National Institute of Mental Health MH089983 and MH089924.
 Support for developing statistical analyses (RTS, TDS) was provided by a seed grant by the Center for Biomedical Computing and Image Analysis (CBICA) at Penn.
 Support for this project was provided by T32MH065218 (SNV); R01MH107235 (RCG, REG, RTS), R01MH107703 (TDS, RTS), and R01NS085211 (RTS).
  
 {\it Conflict of Interest}: None declared.

\nocite{dawid_matrix-variate_1981,van_der_vaart_asymptotic_2000}

\bibliographystyle{apalike}
\bibliography{MyLibrary}

\pagebreak
\begin{center}
\textbf{\large Supplementary material for, ``Faster family-wise error control for neuroimaging with a parametric bootstrap"}
\end{center}
\setcounter{equation}{0}
\setcounter{figure}{0}
\setcounter{table}{0}
\setcounter{page}{1}
\makeatletter
\renewcommand{\theequation}{S\arabic{equation}}
\renewcommand{\thefigure}{S\arabic{figure}}
\renewcommand{\bibnumfmt}[1]{[S#1]}
\renewcommand{\citenumfont}[1]{S#1}

The code to execute the simulations and analyses for this manuscript is available at \url{https://bitbucket.org/simonvandekar/param-boot}.

\section{Supplementary proofs}

The proof of Theorem \ref{thm:JointFdistribution} requires defining a matrix normal distribution and identifying some useful properties of matrix-variate random variables \citep{dawid_matrix-variate_1981}.

\begin{theorem}[Properties of matrix-variate random variables]
\label{thm:matrixprops}
The following are properties for matrix-variate normal random variables.
\begin{enumerate}
\item Let the $n \times p$ matrix $Z$ have independent standard normal entries, Then the matrix $A + CZB \sim \mathcal{MN}(A, CC^T, B^TB)$, is matrix-variate normal with mean matrix $A$, row covariance matrix $CC^T$, and column covariance $B^TB$.
\item Let the $n \times p$ matrix $X \sim \mathcal{MN}(0, I_{n\times n}, \Sigma)$, then $X^TX \sim \mathcal{W}_p( n, \Sigma)$. If $n<p$ then $\mathcal{W}_p( n, \Sigma)$ is a singular Wishart distribution.
\end{enumerate}

The following are properties of matrix-variate random variables

\begin{enumerate}
\item Let $Z$ be an $n\times p$ matrix-variate random variable. For positive semi-definite matrices $\Psi$ ($n\times n$) and $\Phi$ ($p \times p$), if the row covariance $\mathrm{cov}(Z^T_i) = \Psi \Phi_{ii}$ and the column covariance $\mathrm{cov}(Z_j) = \Phi \Psi_{ii}$, then we write $\mathrm{cov}(Z) = ( \Psi, \Phi)$.
\item If $Z$ is an $n\times p$ matrix-variate random variable with $\mathrm{cov}(Z) = ( \Psi, \Phi)$, then $\mathrm{cov}(\mathrm{vec}(Z)) = \Phi \otimes \Psi$.
\end{enumerate}

\end{theorem}

\begin{proof}[Proof of Theorem \ref{thm:JointFdistribution}]

For the first property write
\begin{equation}
\label{eq:projmat}
\begin{array}{rl}
(R_{X_0} - R_X) & = A A^T \\
R_{X} & = B B^T,
\end{array}
\end{equation}
where $A$ is an $n \times m_1$ matrix and $B$ is and $n \times (n-m)$ matrix both with with orthonormal columns.
Then, under the assumption \eqref{eq:firstmoment}, $A^T X_0 \alpha= (A^T A) A^T X_0 \alpha = A^T (R_{X_0} - R_X) X_0 \alpha = 0$, since $X_0$ is orthogonal to the column space of $R_{X_0}$ and $R_X$.
Then with normal errors \eqref{eq:normality}, Theorem \ref{thm:matrixprops} implies $A^T Y \sim \mathcal{MN}(0, I_{m_1\times m_1}, \Psi)$.
Similarly, $B^T X_0\alpha = 0$ and we obtain $B^T Y \sim \mathcal{MN}(0, I_{(n-m) \times (n-m)}, \Psi)$.
Then equation \eqref{eq:wisharts} follows from Theorem \ref{thm:matrixprops}.

For the proof of the second property let $A$ and $B$ be as defined in \eqref{eq:projmat}.
To prove the convergence of $m_1F_{n}$, we will invoke the central limit theorem for $A^TY$.
To do this we use the Cram\'er-Wold device \citep{van_der_vaart_asymptotic_2000} to prove that $A^TY$ converges in law to a $\mathcal{MN}(0, I_{m_1 \times m_1}, \Psi)$ distribution.

The $m_1 \times V$ matrix-variate random variable $A^TY$ has $v$th row covariance $\text{cov}(A^T Y_v) = \Psi_{v,v}I_{m_1 \times m_1}$ and $i$th column covariance $\text{cov}(A^T_i Y) = \Psi$.
Theorem \ref{thm:matrixprops} implies that the vectorized version has covariance $\text{cov}(\text{vec}(A^TY)) = \Psi \otimes I_{m_1 \times m_1}$.
Using the Cram\'er-Wold device we need only prove that for any vector $t$, 
\begin{equation}
\label{eq:asympnorm}
t^T\text{vec}(A^TY) \rightarrow_{L} \mathcal{N}(0, t^T (\Psi \otimes I_{m_1 \times m_1}) t ).
\end{equation}

Assumption \eqref{eq:firstmoment} implies $t^T \text{vec}(A^T \E Y) = 0$, by the same argument as above.
Assumption \eqref{eq:secondmoment} implies $t^T (\Psi \otimes I_{m_1 \times m_1}) t< \infty$. So, by the central limit theorem
\eqref{eq:asympnorm} holds, which implies $A^T Y \rightarrow_{L} \mathcal{MN}(0, I_{m_1 \times m_1}, \Psi)$ by the Cram\'er-Wold device.
From there, the continuous mapping theorem gives $\Phi^2 \text{diag}\{Y^T A A^T Y\} \rightarrow_L \text{diag}\{\mathcal{W}_V(m_1, \Sigma )\}$.
 
 For the denominator let $B$ be as defined in \eqref{eq:projmat}.
 Then, under $Y_{iv} \independent Y_{jv}$ for all $i\ne j$ and letting $W_v = B^T Y_v$, 
 \[
 \text{cov}(W_v) = B^T \text{cov}(Y_v) B = \Psi_{v,v} I_{(n-m)\times (n-m)}.
 \]
Assumption \eqref{eq:firstmoment} means $\E W_v = 0$ by the same argument as for $A^T \E Y$, so
\[
\frac{1}{n-m} \E Y^T_v R_X Y_v =  \E \frac{1}{n-m} W_{v}^T W_v = \frac{1}{n-m} \sum_{i=1}^{n-m} \E W_{iv}^2 = \Psi_{v,v}.
\]
By the weak law of large numbers $Y^T_v R_X Y_v/(n-m) \rightarrow_{P} \Psi_{v,v}$.
 The convergence of $m_1 F_{n} \rightarrow_L \text{diag}\{\mathcal{W}_V(m_1, \Sigma)\}$ follows by Slutsky's theorem \citep{van_der_vaart_asymptotic_2000}.

\end{proof}

 \renewcommand*{\proofname}{Proof}

The joint CDFs for the numerators can be used to estimate the asymptotic joint CDF of the transformed statistics \eqref{eq:transformed}.
We use Monte Carlo simulations to generate the numerators using the distributions given in Theorem \ref{thm:JointFdistribution} to estimate the null distribution of $Z_{vn}$ given in \eqref{eq:transformed}.

To show that the PBJ procedure guarantees asymptotic control of the FWER we must show that the joint distributions satisfy the null domination condition of Definition \ref{def:and}. When null domination holds then Theorem \ref{thm:aFWERc} guarantees asymptotic control of the FWER.
\begin{theorem}[Null domination]
\label{thm:nulldom}
Let $F_{n} = (F_{1n}, \ldots, F_{Vn})$, and $Z \sim Q_0 = \mathrm{diag}\{ \mathcal{W}_V(m_1, \Sigma)\}$. Let, $g_n(x) = \Phi_n^{-1}(\Phi(x))$ where $\Phi$ and $\Phi_n$ are as defined in \eqref{eq:transformed}. Then, the joint distribution $Q_n$ of $Z_{n}$, defined by the transformation \eqref{eq:transformed}, is asymptotically dominated by the joint distribution $Q_0$ of $Z$.
\end{theorem}

The following lemma is used to prove Theorem \ref{thm:nulldom} and is presented here without proof.

\begin{theorem}{lemma}
\label{lem:converge}
Let $f_n:\R \mapsto [0,1]$ converge uniformly to a continuous function $f$ and $g_n: \R \mapsto \R$ converge pointwise to $g$.
Then $f_n(g_n(x))$ converges pointwise to $f(g(x))$.
\end{theorem}

\begin{proof}
Let $x$ and $\epsilon$ be given.
Because $f$ is continuous there exists $\delta$ such that $\lvert f(y) - f(g(x)) \rvert <\epsilon/2$
for all $y$ such that $\lvert g(x) - y \rvert < \delta$.
Choose $N_1 = N_1(\epsilon, x) $ such that for all $n\ge N$, $\lvert g_n(x) -g(x) \rvert < \delta$, which is possible due to the pointwise convergence of $g_n$.
Because $f_n$ converges uniformly, there exists $N_2 = N_2(\epsilon)$ such that $\lvert f_n(y) - f(y) \rvert < \epsilon/2$ for all $y \in \R$.
Thus, it follows that for all $n \ge N = \max\{N_1, N_2\}$.
\[
\lvert f_n(g_n(x)) - f(g(x)) \rvert \le \lvert f(g_n(x)) - f(g(x)) \rvert + \lvert f_n(g_n(x)) - f(g_n(x)) \rvert
< \epsilon
\]

\end{proof}

\begin{proof}[Proof of Theorem \ref{thm:nulldom}]
For any $V$ dimensional vector of random variables $Z$ let $F_{ Z }(x) = \P(\max_{v\le V} Z_{j}  < x)$.
We will show that $F_{Z_n}(x) \rightarrow F_{Z}(x)$, for all $x$ as $n\to \infty$.
This implies that the null domination condition holds because then $\limsup_{n\to\infty} 1-F_{Z_n}(x)  \le 1-F_{Z}(x)$.

First note that $F_n$ are continuous random variables, so taking the maximum of $F_n$ is a continuous function.
The continuous mapping theorem implies that 
\begin{equation}\label{eq:convergence}
\max_{v} m_1F_{vn} \rightarrow_L \max_{v} Z_v,
\end{equation}
because $m_1 F_n \to_L Z$ by Theorem \ref{thm:JointFdistribution}.

$g_n$ is monotone in $x$, which implies $F_{Z_n}(x) = F_{F_n}\{g_n(x)\}$ because $Z_{vn} = g_n^{-1}(F_{vn})$. Thus for any $n$, 
\[
F_{Z_n}(x) - F_{Z}(x) = F_{F_n}\{g_n(x)\} - F_Z(x).
\]
$F_{F_n}(m_1^{-1}x) - F_Z(x)$ converges to zero pointwise by \eqref{eq:convergence} and since $F_{Z}$ is continuous then convergence in law implies uniform convergence to $F_Z(x)$ \citep[Lemma 2.11]{van_der_vaart_asymptotic_2000}. Since $g_n(x) \to m_1^{-1}x$, then uniform convergence of $F_{F_n}(m_1^{-1}x)$ to $ F_Z(x)$ and the continuity of $F_Z(x)$ imply $F_{F_n}\{g_n(x)\} - F_Z(x) \to 0$ by Lemma \ref{lem:converge}.

\end{proof}

The assumption of convergence  of $Z_n$ to $Z$ in Theorem \ref{thm:nulldom} holds for many statistics of interest by the central limit theorem.
When Theorem \ref{thm:nulldom} holds then, the following theorem from \citet[p. 205]{dudoit_multiple_2008} ensures asymptotic control of the FWER.
\begin{theorem}[Asymptotic control of the FWER by step-down procedure]
\label{thm:aFWERc}
Let $Z_{(1)n} < \ldots < Z_{(V)n}$ denote the ordered test statistics, and $H_{(1)}, \ldots H_{(V)}$ their associated hypotheses.
Assume that the distribution, $Q_n$, for the test statistics $Z_{(v)n}$ is dominated asymptotically by $Q_0$ and that $Z = Z_{(1)}, \ldots, Z_{(V)} \sim Q_0$.
 For a given level $\alpha$ define the thresholds $C_{vn}$ as the smallest value that satisfies $\P\left(\max_{k \le v} Z_{(k)} <  C_{vn} \right) = 1-\alpha.$
Where $C_{vn}$ depends on the sample through the order of the test statistics $Z_{(v)n}$.
Let $H_0$ denote the set of true null hypotheses and define the number of type 1 errors as $E_n = \sum_{v=1}^V I\left( Z_{(v)n} > C_{vn}, H_{(v)} \in H_{0} \right)$.

Then the PBJ procedure provides asymptotic control of the FWER at level $\alpha$, 
\[
\limsup_{n\to \infty} \P(E_n >0) \le \alpha.
\]
\end{theorem}
The proof of Theorem \ref{thm:aFWERc} is given in \citet[p. 206]{dudoit_multiple_2008}. Propostion 5.5 of \citet[p. 211]{dudoit_multiple_2008} gives the adjusted \emph{p}-values used in Procedure \ref{proc:bootSD}.
The proof of asymptotic FWER control for Procedure \ref{proc:bootSS} is implied by Theorem \ref{thm:aFWERc}, because the single step procedure uses the single most conservative threshold, so leads to more conservative inference.

The joint distribution $Q_0$ is unavailable in practice and must be estimated from the sample.
Theorem 5.12 of \citet[p. 228]{dudoit_multiple_2008} states that if our estimate $\hat{Q}_{0}= Q_{0}(\hat \Sigma)$ converges in probability to $Q_0$ then the thresholds $C_{vn}$ in Theorem \ref{thm:aFWERc} are consistent, which gives consistent adjusted \emph{p}-values by the continuous mapping theorem.
Because our estimate $\hat \Sigma$ in \eqref{eq:sigmahat} is consistent, then $\hat{Q}_{0}$ converges in probability to $Q_0$. Thus, using this estimate of $\Sigma$ gives valid asymptotic inference using the PBJ.

\subsection{Supplementary simulation analyses}
We also ran the simulation analyses with 3mm and 9mm smoothing.
The parameter image for each simulation was smoothed with the same kernel as the imaging data.
The results are presented in Figures \ref{fig:sm3mm} and \ref{fig:sm9mm}.

\begin{knitrout}
\definecolor{shadecolor}{rgb}{0.969, 0.969, 0.969}\color{fgcolor}\begin{figure}[p!]

{\centering \includegraphics[width=0.8\maxwidth]{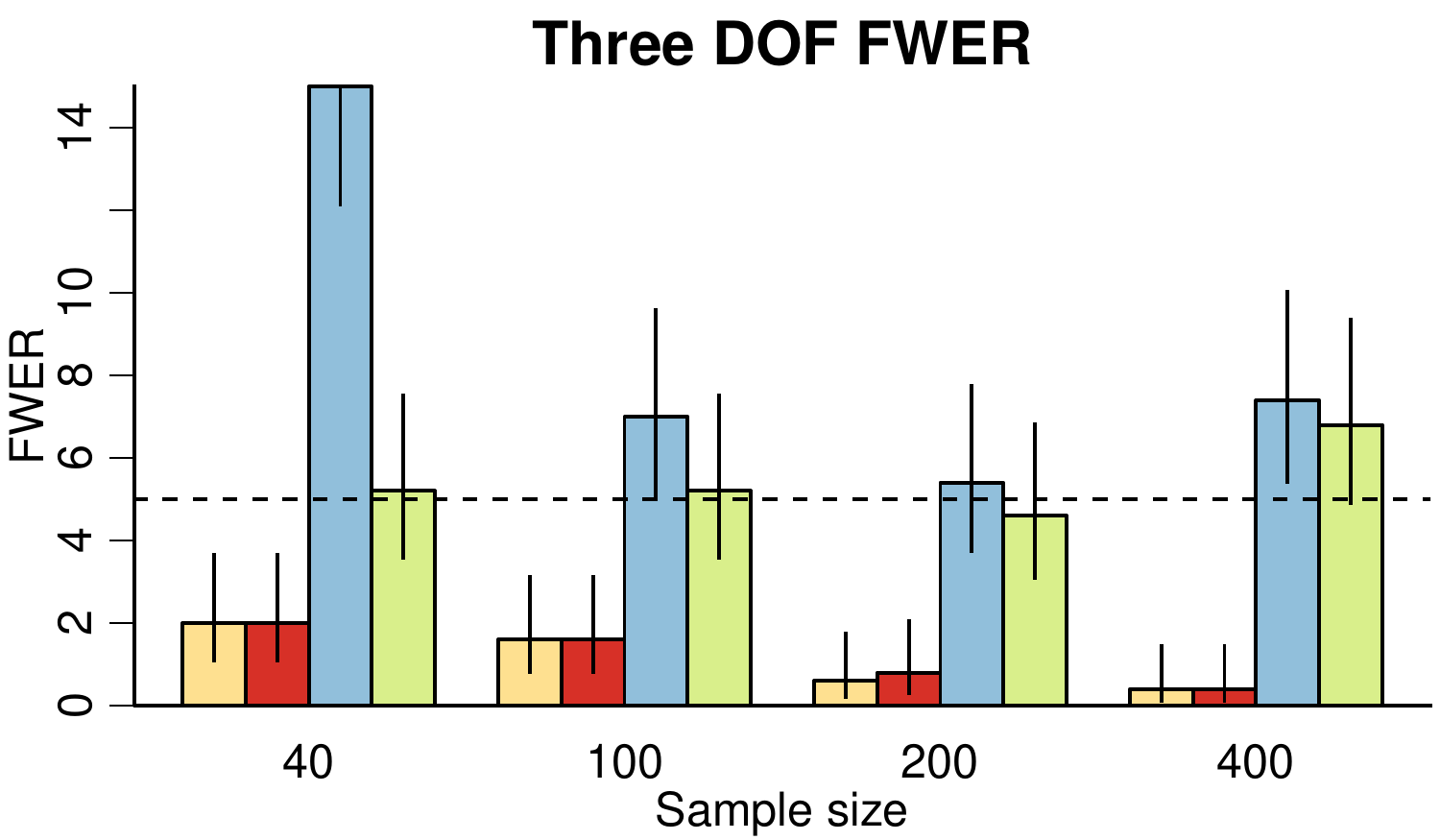} 
\includegraphics[width=0.8\maxwidth]{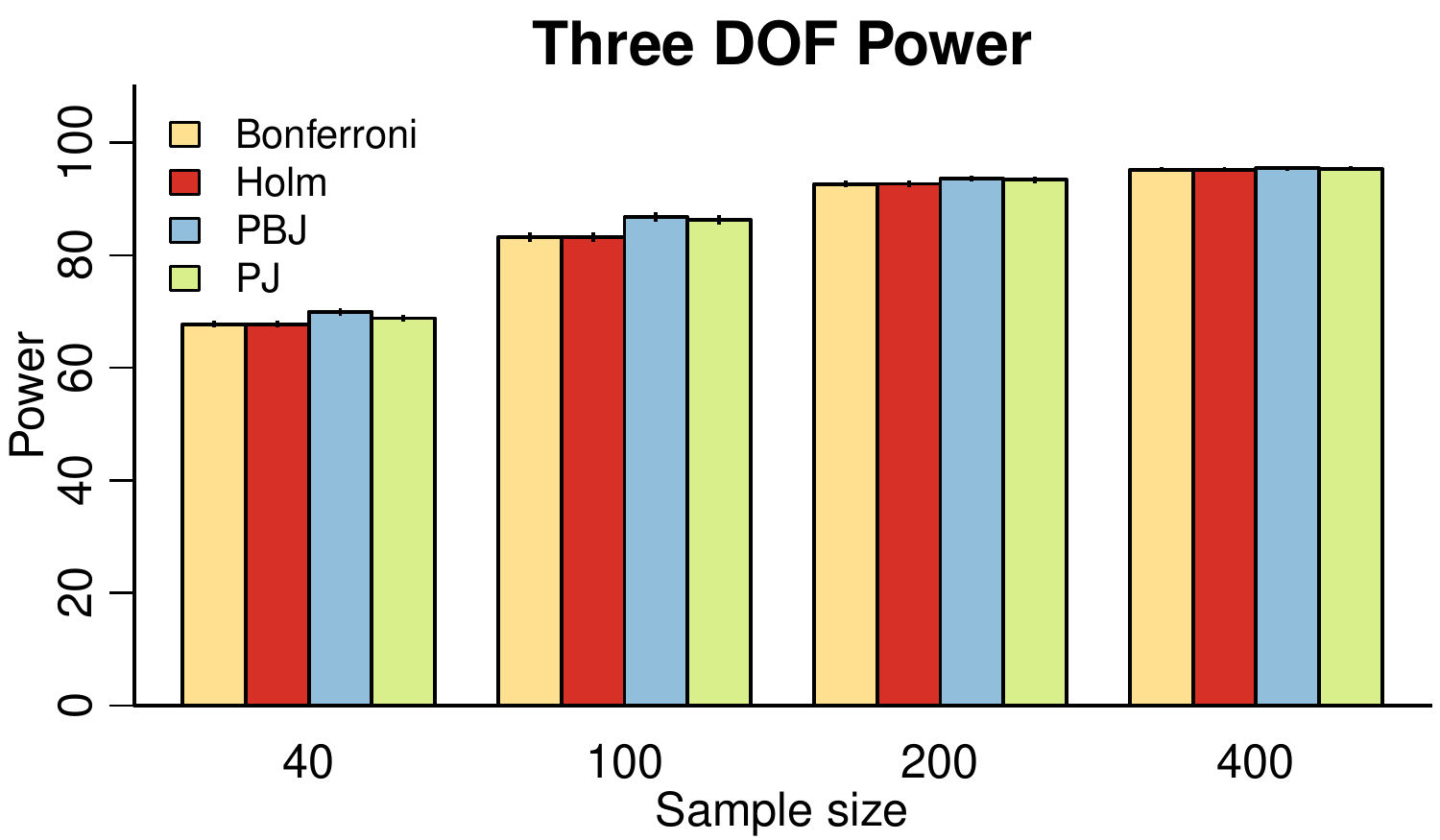} 

}

\caption{FWER and power for simulations with a Gaussian smoothing kernel of FWHM=3mm for the F-statistic of \eqref{eq:threedof} on three degrees of freedom (DOF) for the CBF image.}\label{fig:sm3mm}
\end{figure}

\end{knitrout}

\begin{knitrout}
\definecolor{shadecolor}{rgb}{0.969, 0.969, 0.969}\color{fgcolor}\begin{figure}[p!]

{\centering \includegraphics[width=0.8\maxwidth]{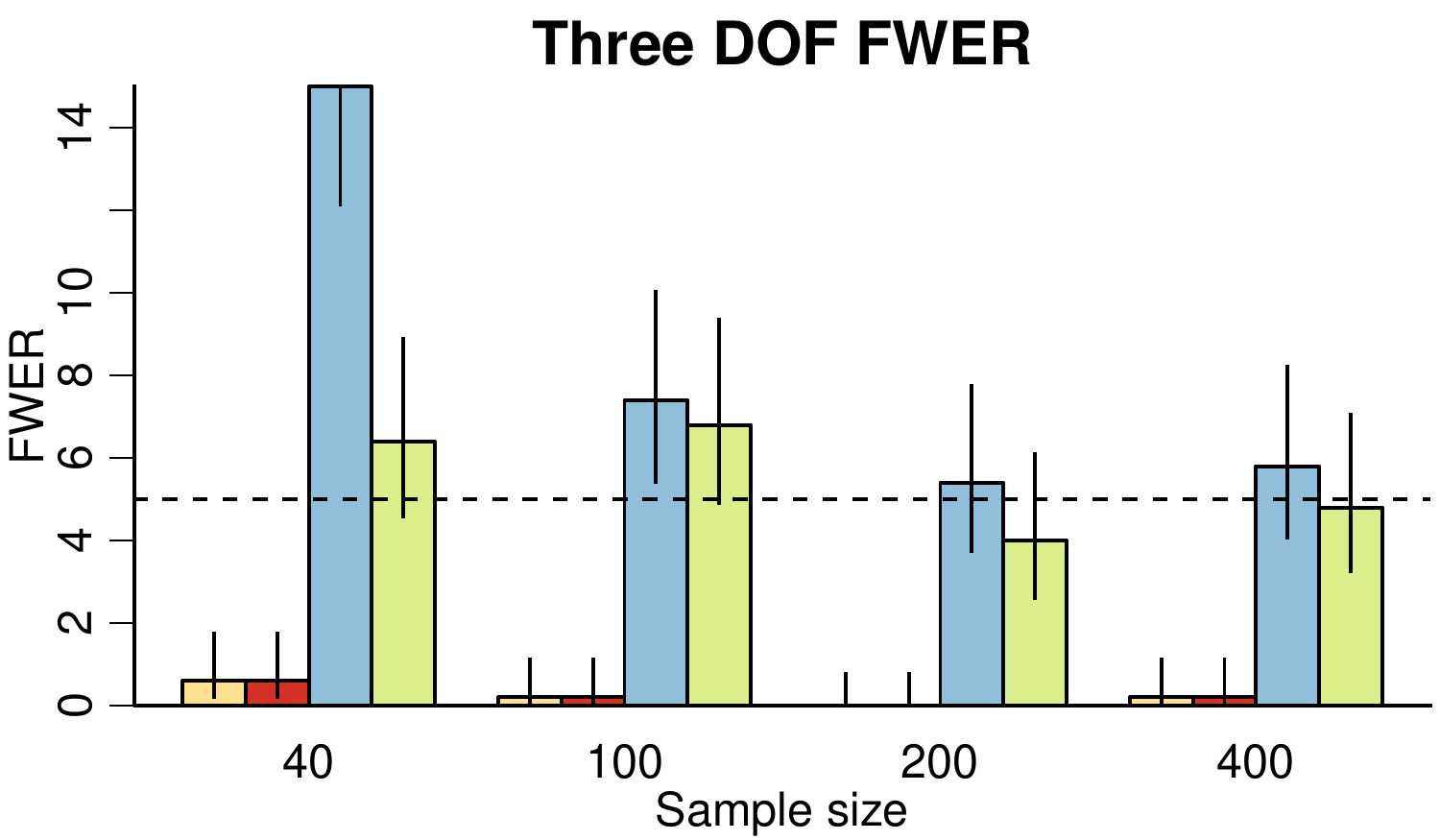} 
\includegraphics[width=0.8\maxwidth]{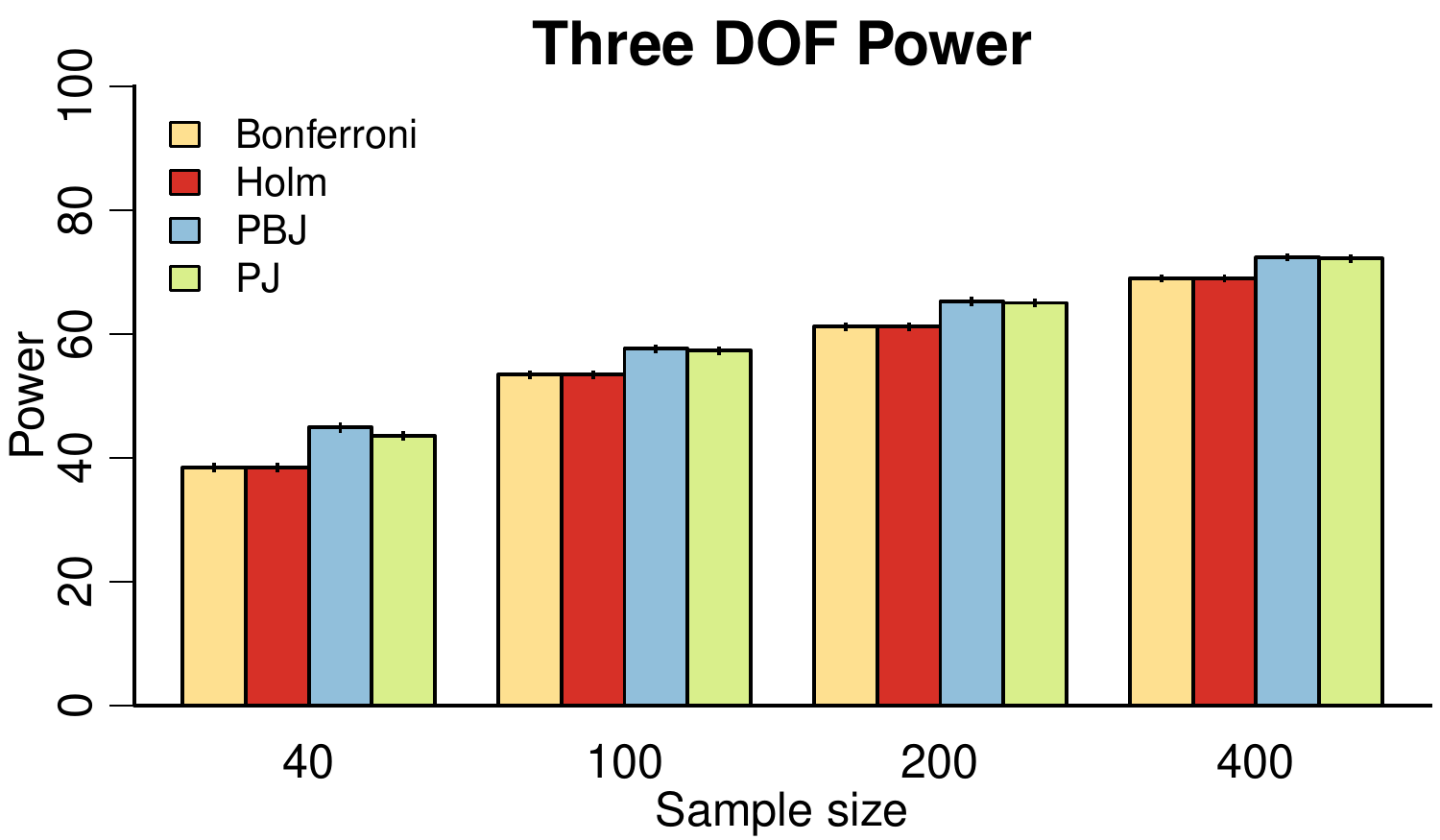} 

}

\caption{FWER and power for simulations with a Gaussian smoothing kernel of FWHM=9mm for the F-statistic of \eqref{eq:threedof} on three degrees of freedom (DOF) for the CBF image.}\label{fig:sm9mm}
\end{figure}

\end{knitrout}

Figure \ref{fig:CBFk10} shows FWER controlled results fitting a fixed degree spline model with 10 knots. We presented results with 5 knots in the manuscript.
The permutation procedure yielded no significant results because in the untransformed data there are voxels whose test statistics have heavily skewed null distributions.
When taking the maximum across the image this leads to very conservative results.
\begin{figure}
\centering\includegraphics[width=\maxwidth]{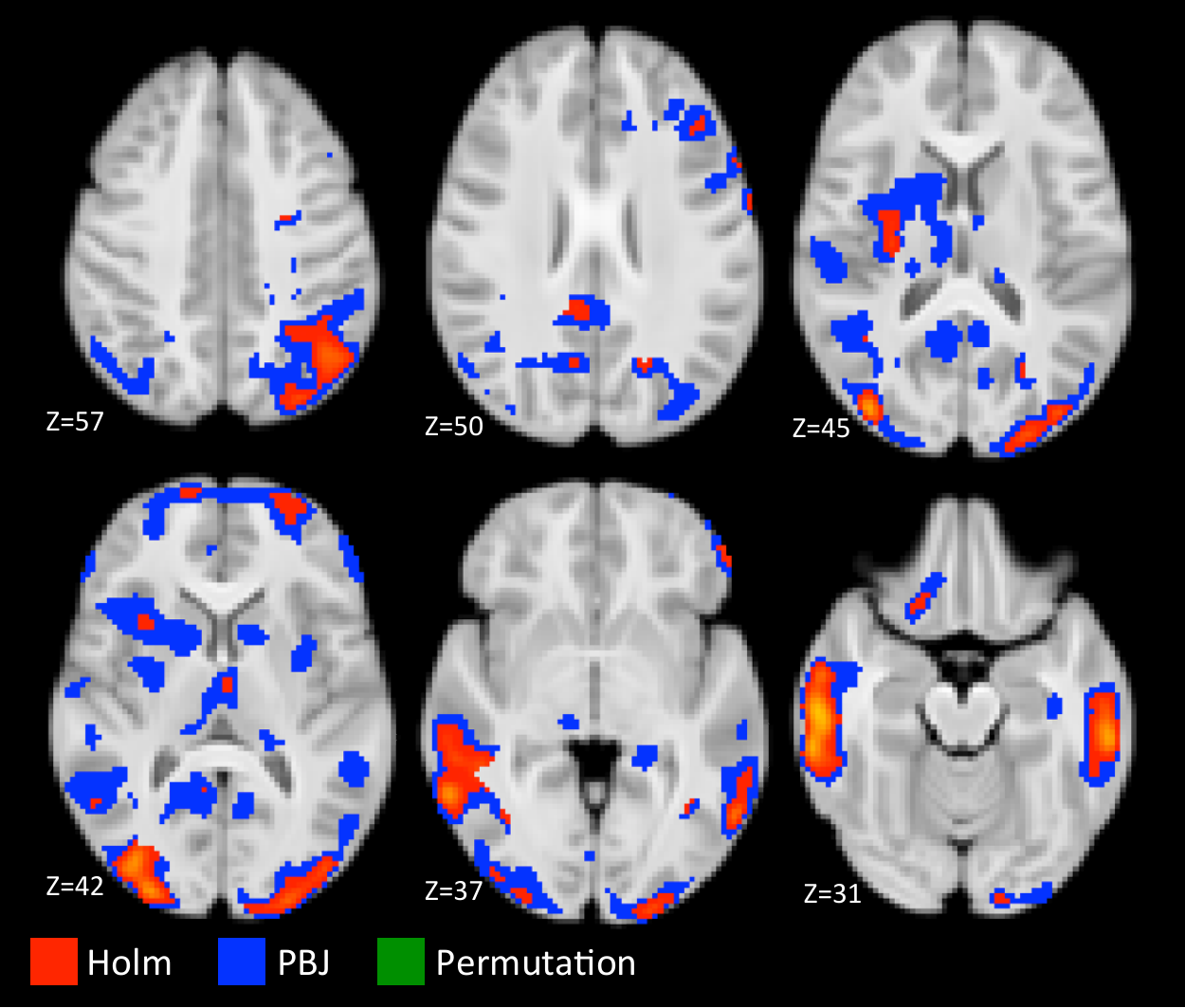} 
\caption{FWER controlled results at $\alpha=0.05$ for Holm (red), PBJ single-step (blue), and PJ single-step (green) for the spline model fit with 10 knots. Color scale is $-\log_{10}(p)$ for the adjusted \emph{p}-values and shows results greater than 1.3. The overlay order is the PBJ, Holm, and PJ procedures, so that green indicates regions where all three regions reject the null, red and green indicate regions where Holm and PJ reject, and the union of all colors is where PBJ rejects. Blue indicates locations where only the PBJ procedure rejects.}\label{fig:CBFk10}
\end{figure}

\end{document}